\documentclass[12pt,letterpaper,leqno]{article} %v2.0
\usepackage{amsmath,amsfonts,amssymb,amsthm,enumerate,enumitem,multirow,array,graphicx,lscape,subcaption,lastpage,color,epstopdf,setspace,mathrsfs,listings,apptools,natbib,algorithmic,algorithm,cite}

\usepackage[letterpaper,margin=1in]{geometry}
\usepackage[bookmarksdepth=2,colorlinks=false]{hyperref}
\usepackage[OT1]{fontenc}
\usepackage{titlesec,titling}
\pretitle{\begin{center}\Large\bfseries}
\posttitle{\par\end{center}}
\preauthor{\begin{center}\lineskip 0.5em \scshape\begin{tabular}[t]{c}}
\postauthor{\end{tabular}\par\end{center}}
\predate{\hphantom}
\postdate{\vskip -1em}

\providecommand{\keywords}[1]{\par{\scshape Keywords:} #1}
\providecommand{\jelcodes}[1]{\par{\scshape JEL Codes:} #1}
\titleformat{\section}[runin]{\normalfont\bfseries}{\relax\thesection.~}{0pt}{}[.]
\AtAppendix{\titleformat{\section}[block]{\normalfont\scshape\filcenter}{Appendix~\thesection.~}{0pt}{}[]}
\titleformat{\subsection}[runin]{\normalfont\IfAppendix{\bfseries}{\itshape}}{\textup{\thesubsection.}~}{0pt}{}[.]
\titleformat{\subsubsection}[runin]{\normalfont\itshape}{\textup{\thesubsubsection.}~}{0pt}{}[.]
\titlespacing{\section}{\parindent}{1em}{1em}
\titlespacing{\subsection}{\parindent}{1em}{1em}
\titlespacing{\subsubsection}{\parindent}{1em}{1em}
\newtheoremstyle{imsthm}{1em}{1em}{\itshape}{\parindent}{\scshape}{.}{.5em}{}
\AtAppendix{\numberwithin{thm}{section}}
\usepackage[font=small,labelfont=sc]{caption}
\theoremstyle{imsthm}
\newtheorem{thm}{Theorem}
\newtheorem{lmm}[thm]{Lemma}
\newtheorem{pro}[thm]{Proposition}
\newtheorem{crl}[thm]{Corollary}
\newtheoremstyle{imsdfn}{1em}{1em}{\normalfont}{\parindent}{\scshape}{.}{.5em}{}
\theoremstyle{imsdfn}

\AtAppendix{\numberwithin{equation}{section}}
\usepackage{natbib}
\setlist{noitemsep}

\DeclareMathOperator*{\argmax}{argmax}

\usepackage{tikz}
\usetikzlibrary{calc,fit,positioning}
\usetikzlibrary{decorations.pathreplacing}
\usepackage{epigraph}

\begin{document}

\title{Disclosure Games with Large Evidence Spaces}
\author{Shaofei Jiang\thanks{Department of Economics, the University of Texas at Austin (email: shaofeij@utexas.edu). I thank V Bhaskar, Barton Lipman, Maxwell Stinchcombe, Caroline Thomas, and seminar participants at UT Austin, Stony Brook Game Theory Conference, Midwest Economic Theory Conference, and Texas Economic Theory Camp for helpful comments. All remaining errors are my own. This paper was previously titled ``Lying by omission.''}}

\maketitle
\thispagestyle{empty}

\onehalfspacing

\begin{abstract}
We study a disclosure game with a large evidence space. There is an unknown binary state. A sender observes a sequence of binary signals about the state and discloses a left truncation of the sequence to a receiver in order to convince him that the state is good. We focus on truth-leaning equilibria (cf. \citet{hart_kremer_perry_2017}), where the sender discloses truthfully when doing so is optimal, and the receiver takes off-path disclosure at face value. In equilibrium, seemingly sub-optimal truncations are disclosed, and the disclosure contains the longest truncation that yields the maximal difference between the number of good and bad signals. We also study a general framework of disclosure games which is compatible with large evidence spaces, a wide range of disclosure technologies, and finitely many states. We characterize the unique equilibrium value function of the sender and propose a method to construct equilibria for a broad class of games.

\jelcodes{C72, D82, D83}
\end{abstract}

\epigraph{\itshape Seldom, very seldom does complete truth belong to any human disclosure; seldom can it happen that something is not a little disguised, or a little mistaken; but where, as in this case, though the conduct is mistaken, the feelings are not, it may not be very material.}{---Jane Austen, \itshape Emma}

\section{Introduction} \label{sec1}

Eliciting information from a biased party is challenging. Verifiability of information alleviates the problem by eliminating lies from disclosure, but rarely can we exclude omission of information. For example, new hedge funds operate on internal capital in an incubation period before being offered to outside investors. Potential investors do not observe when the incubation period starts. When a hedge fund turns full fledged, a part of its trading record from the incubation period is disclosed as a guidance to future returns. Since investors do not know the starting date of the incubation period, the hedge fund manager can truncate the trading record from the left and disclose only the hedge fund's most recent returns in order to convince investors that the hedge fund's average return is high. The manager may prefer doing so to disclosing the entire trading record if the hedge fund has relatively more low returns at the beginning of the incubation period. The investors, on the other hand, do not take the disclosed trading record at face value and may suspect that the fund's actual return is lower than the disclosed part appears.

This is an example of a \emph{disclosure game}. A sender has a piece of hard evidence about an unknown state of the world and tries to influence the posterior belief of a receiver by strategically disclosing her evidence with possible omission. A \emph{technology of disclosure} specifies, for every possible evidence endowment, the set of messages that the sender can disclose. The technology is common knowledge, and the sender cannot make ex ante commitment to truthfully disclosing her evidence without omission. Hence, the receiver is skeptical about the disclosed message and believes that unfavorable evidence may be concealed by the sender.

There are two distinct features of the hedge fund example. First, the amount of evidence is potentially unbounded. Regardless of how long a trading record the manager discloses, it is always possible that the actual incubation period is longer, so the investors cannot rule out that some low returns in the early stage of the incubation period are concealed by the fund manager. Therefore, skepticism about the disclosed trading record remains. This feature is common in other environments. For example, when disclosing proprietary information, due to the large variety of metrics that a firm can disclose, the market is never certain that all relevant information is disclosed regardless of how much information the firm discloses.

Second, the returns have a sequential structure, and the technology of disclosure permits only \emph{left censoring}. That is, the fund manager can misreport the starting date of the incubation period and disclose a left truncation of the trading record, but it is often impossible for her to, say, disclose returns only for periods in which the fund performs well. Different scenarios may have different technologies of disclosure. A commonly studied technology of disclosure allows the sender to disclose any subset of her evidence. For example, a lawyer can present any evidence that she believes can help her case and try to exclude other evidence from being presented in the court, and a car dealer can focus on the selling points of a car and avoid mentioning where the car falls behind.

We model the interaction between the hedge fund manager and the investors as a disclosure game with left censoring technology. The sender is the hedge fund manager, and the receiver is a representative investor. The manager is endowed with a trading record, which is a dated sequence of binary signals, high or low. Each signal indicates the return in a week in the incubation period. The trading record is informative of whether the hedge fund's trading strategy is a success. If the hedge fund is a success, it has frequent high returns during the incubation period, whereas in the case of a failure, high returns occur with lower frequency. The manager discloses a message, which is a left truncation of the trading record, to the investor. The amount of investment the hedge fund receives is an increasing function of the investor's posterior belief on the fund being a success. Therefore, the manager's objective is to convince the investor that the hedge fund is a success.

\subsection{Preview of results}

In order to solve this disclosure game with left censoring technology, we study a general framework of disclosure games which is compatible with large evidence spaces, a wide range of disclosure technologies, and any finite states. We focus on perfect Bayesian equilibria satisfying the truth-leaning refinement by \citet{hart_kremer_perry_2017}. That is, the sender truthfully discloses her evidence endowment whenever it is optimal to do so, and the receiver regards any off-path message as the truthful disclosure of the sender's evidence. This refinement follows the intuition that there is a ``slight inherent advantage'' for the sender to disclose truthfully. Our characterization of truth-leaning equilibrium consists of three parts. First, we characterize the unique equilibrium value function of the sender and propose an induction process to solve it. This implies that the receiver's strategy is the same in all truth-leaning equilibria. Second, we give another induction process to construct the equilibrium strategies and beliefs. Third, we prove a sufficient condition for equilibrium existence. The condition is widely satisfied, so for a broad class of disclosure games, the two induction processes yield truth-leaning equilibria.

Applying these results, we show that the disclosure game with left censoring technology between the hedge fund manager and the investor has the following equilibrium features.

First, initial low returns are disclosed in equilibrium. This is somewhat surprising, since the sender can omit these initial low returns and report a shorter truncation which contains fewer low returns. In terms of face value, disclosing initial low returns is never optimal. In fact, in our setting, the posterior belief on whether the hedge fund is a success depends only on the difference between the number of high returns and the number of low returns in the trading record. Hence, a seemingly optimal strategy for the manager is to disclose a \emph{maximal difference truncation} of the trading record, i.e., a left truncation which maximizes the difference between the number of high and low returns among all truncations of the trading record. However, in equilibrium, the manager may disclose a longer truncation, and the investor is not influenced by the redundant unfavorable evidence at the beginning of the disclosed truncation.

To understand the intuition, consider a trading record such that there are more low returns than high returns in every nonempty left truncation of it. For such trading record, the unique maximal difference truncation is the empty history--that is, truncating away all returns is better than disclosing any nonempty truncation which contains more low returns than high returns in terms of face value. However, suppose that the sender discloses the empty history whenever she has one such trading record. Then observing the empty history, the receiver's posterior belief on the hedge fund being a success is so low that some sender types whose trading records have just a few more low returns than high returns may find it profitable to deviate to disclosing truthfully, since the receiver takes off-path disclosure at face value. Hence, in a truth-leaning equilibrium, the sender discloses different messages, including seemingly sub-optimal ones, and those seemingly sub-optimal messages lead to the same receiver action as the maximal difference truncation.

Second, when there are more than one maximal difference truncations, the equilibrium message always contains the longest maximal difference truncation (it may be even longer, as is explained above). Suppose that in an equilibrium, a sender has two maximal difference truncations and discloses the shorter one with positive probability. Then the investor's posterior belief after observing the shorter truncation must be strictly lower than that after observing the longer truncation, because the longer truncation is not feasible for some sender types which have shorter trading records. Hence, disclosing the longer truncation is a profitable deviation for the sender. In equilibrium, sender types with longer trading records also disclose longer truncations, so that sender types with shorter trading records who cannot disclose longer truncations pool on disclosing shorter truncations.

\subsection{Related literature}

This paper is related to an extensive literature on verifiable disclosure games with specific evidence structures and technologies of disclosure. \citet{grossman_1981} and \citet{milgrom_1981} show unraveling, i.e., full revelation of evidence in equilibrium, if the sender can prove any true claim. More recent work studies the case where the amount of the sender's evidence is uncertain, and the sender cannot prove her lack of evidence. For example, \citet{dye_1985} studies a scenario where the sender may have no evidence, and a sender with some evidence can either disclose truthfully or pretend to have no evidence; \citet{shin_2003} studies a game where evidence consists of binary signals, and the sender can disclose any subset of the signals to the receiver; \citet{wolinsky_2003} studies a case where the sender's evidence is a real number, and the sender can disclose any number bounded by her evidence to the receiver. A common theme of these papers is that unraveling fails when the amount of the sender's evidence is uncertain, and the sender cannot prove her lack of evidence. In the current paper, we study a disclosure game with an infinite evidence space and left censoring technology and find a similar result. Since the receiver is uncertain about the the length of the incubation period, the sender can pretend to possess a shorter trading record rather than having a trading record that has more low returns.

Another feature of our model of left censoring disclosure is that the sender uses seemingly sub-optimal messages in equilibrium. Our explanation complements previous studies which explain similar observations. \citet{dziuda_2011} models a disclosure game with strategic and honest senders. The honest type always reveals the evidence without omission, so her message is discounted less by the receiver. The strategic type therefore has an incentive to mimic the honest type by reporting some seemingly unfavorable evidence. In our model, we only have the strategic type. \citet{guttman_kremer_skrzypacz_2014} study the dynamic disclosure of firm performance, and show that due to dynamic considerations, later reports are interpreted as more favorable although the timing of obtaining signals is independent of the firm's value. Our model does not have dynamic considerations. Moreover, it is worth pointing out that the informativeness of each signal in the trading record is the same in our model, so our result does not arise because the receiver values earlier signals differently.

This paper also makes methodological contributions to solving general disclosure games. We model technologies of disclosure as preorders on an evidence space. This representation is consistent with the standing assumptions in the literature on verifiable disclosure (e.g., \citet{bull_watson_2007}, and \citet{ben-porath_lipman_2012}) but is more general in that we allow the evidence space to be infinite. We characterize the sender's equilibrium value function in any disclosure game and construct equilibria for a broad class of games. Two papers that are the closest to this one are \citet{hart_kremer_perry_2017} and \citet{rappoport_2017}. \citet{hart_kremer_perry_2017} propose the truth-leaning refinement that we use in this paper. They show that any truth-leaning equilibrium is receiver optimal. That is, truth-leaning equilibria yield the same ex post payoffs as the optimal mechanism where the receiver commits to a reward scheme ex ante. However, they do not characterize truth-leaning equilibria or the optimal mechanism. \citet{rappoport_2017} characterizes the sender's payoffs in the receiver optimal equilibrium. This characterization is similar to ours, but our methods are very different. \citet{rappoport_2017} essentially solves a mechanism design problem, and as a result, the characterization gives little insight on equilibrium strategies without commitment. Moreover, the algorithm proposed to solve the equilibrium payoffs does not apply to disclosure games when the evidence space is infinite, e.g., the hedge fund example. On the other hand, we prove our characterization of the equilibrium value function by induction, which allows us to construct truth-leaning equilibria even when the evidence space is infinite.

\subsection{Outline of the paper}

The rest of the paper is organized as follows. Section \ref{sec2} presents the general framework of disclosure games. Section \ref{sec3} characterizes truth-leaning equilibria for general disclosure games. Section \ref{sec4} studies the hedge fund example as a disclosure game with left censoring technology. The last section concludes. All proofs are relegated to the appendix.

\section{Disclosure Games} \label{sec2}

There are two stages. Two players, a sender (she) and a receiver (he), move sequentially. Let \(\Omega=\{\omega_1<\omega_2<\dots<\omega_N\}\) be a finite set of states. At the outset of the game, Nature chooses a state \(\omega\in\Omega\) according to a common prior. Let \(\pi_n>0\) be the probability of state \(\omega_n\). Neither player observes the realized state.\footnote{We will assume that the sender's payoff is independent of the realized state, so it does not affect our analysis if the realized state is known to the sender.} In the first stage, Nature chooses a piece of hard evidence, denoted \(e\), according to a state-dependent distribution. The sender privately observes the evidence and sends a message \(m\) to the receiver. In the second stage, the receiver observes the message and takes an action \(a\in\mathbb{R}\).

\subsection{Evidence}\label{sec2.1}

The evidence space \(E\) is exogenously given and contains finite or countably many elements. Let \(F_1,F_2,\dots,F_N\) be distributions on \(E\). In the first stage, the sender privately observes a piece of evidence \(e\in E\), which is referred to as her \emph{evidence endowment}. If the realized state of the world is \(\omega_n\), \(e\) is a random draw from distribution \(F_n\). That is, \(F_n(e)\) is the probability that the sender's evidence endowment is \(e\) conditional on the state \(\omega_n\). Let \(F\) be the unconditional distribution of the sender's evidence endowment, i.e., \(F(e)=\sum_{j=1}^N F_n(e)\pi_n\). We assume that \(F(e)>0\) for all \(e\in E\).

\subsection{Messages and technologies of disclosure}\label{sec2.2}

After observing her evidence endowment \(e\), the sender reports a message \(m\in \mathcal{M}(e)\) to the receiver, where \(\mathcal{M}(e)\) is the set of messages the sender can feasibly disclose given evidence endowment \(e\). The collection of the sender's feasible messages \(\{\mathcal{M}(e)\}_{e\in E}\) describes the \emph{technology of disclosure}. We consider a broad class of technologies satisfying the following assumptions:\footnote{(A1) through (A3) are standard in the literature (see, for example, \citet{bull_watson_2004,ben-porath_lipman_2012,hart_kremer_perry_2017,ben-porath_dekel_lipman_2019}), and they are without loss of generality if \emph{normality} is assumed \citep{lipman_seppi_1995,bull_watson_2007}. A technology of disclosure \(\{\mathcal{M}(e)\}_{e\in E}\) satisfies normality if for all \(e\in E\), there exists \(m_e\in \mathcal{M}(e)\) such that \(m_e\in \mathcal{M}(e')\Rightarrow \mathcal{M}(e)\subset \mathcal{M}(e')\) for all \(e'\in E\). That is, if a sender with evidence \(e\) can distinguish herself from a sender with evidence \(e'\) (by disclosing some message in \(\mathcal{M}(e)\) but not in \(\mathcal{M}(e')\)), then she can do so using a ``maximal message'' \(m_e\). Hence, it is without loss of generality to consider only maximal messages, i.e., \(\{m_e\}_{e\in E}\), which can be mapped onto from the evidence space \(E\). (A4) is a technical assumption. It is automatically satisfied when the evidence space is finite.}

\begin{description}
	\item[(A1)] \(\mathcal{M}(e)\subset E\) for all \(e\in E\);
	\item[(A2)] (Reflexivity) \(e\in \mathcal{M}(e)\) for all \(e\in E\);
	\item[(A3)] (Transitivity) If \(e'\in \mathcal{M}(e)\) and \(e''\in \mathcal{M}(e')\), \(e''\in \mathcal{M}(e)\);
	\item[(A4)] For all sequences \(\{e_k\}_{k=1}^\infty\) in \(E\) such that \(\mathcal{M}(e_1)\supset \mathcal{M}(e_2)\supset\dots\), there exists \(K\geq 1\) such that \(\mathcal{M}(e_k)=\mathcal{M}(e_K)\) for all \(k\geq K\).
\end{description}

One way to interpret these assumptions is to view the sender's evidence endowment \(e\) as her type. The set of feasible messages \(\mathcal{M}(e)\) is the set of types which sender type \(e\) can imitate. (A2) says that the sender can always truthfully report her type. (A3) says that if sender type \(e\) can imitate type \(e'\), and sender type \(e'\) can imitate type \(e''\), then sender type \(e\) is able to imitate type \(e''\). (A4) is a lower bound condition. It says that our model allows for cases where every sender type \(e\) can be imitated by another type \(e'\), but it does not allow for cases where there exists a sender type \(e^\star\) who can keep imitating another type indefinitely. For example, in the hedge fund example, there are infinitely many possible realizations of trading records, and any trading record can be a truncation of a longer record. However, any trading record has a finite length, so given a trading record, there are only finitely many truncations that the sender can feasibly disclose.

Under (A1) through (A3), a technology of disclosure can be equivalently defined as a preorder \(\precsim\) on \(E\) such that \(e'\precsim e\) if and only if \(e'\in \mathcal{M}(e)\).\footnote{A \emph{preorder} \(\precsim\) is a binary relation satisfying reflexivity (\(e\precsim e\) for all \(e\)) and transitivity (\(e''\precsim e'\precsim e\Rightarrow e''\precsim e\)).} An equivalence of (A4) in terms of this preorder is as follows.
\begin{description}
	\item[(A4$^\prime$)] For all sequences \(\{e_k\}_{k=1}^\infty\) in \(E\) such that \(e_1\succsim e_2\succsim\dots\), there exists \(K\geq 1\) such that \(e_K\precsim e_k\) for all \(k\geq K\).
\end{description}

In the remainder of the paper, we shall abstract away the specific structure of the message space. A technology of disclosure is simply a preorder on the evidence space satisfying (A4$^\prime$). Given a technology \(\precsim\), the sender who is endowed with evidence \(e\) reports a message \(m\precsim e\) to the receiver.

\subsection{The receiver's action}

In the second stage, the receiver observes the sender's message \(m\) and takes an action \(a\in\mathbb{R}\).

\subsection{Payoffs}

The receiver's payoff \(u_R(a,\omega)\) depends on both his action and the realized state of the world. The receiver maximizes his expected payoff. We assume that the receiver's payoff is continuous in his action, and given any distribution on the state space, there is a unique action that maximizes the receiver's expected payoff. That is, for all \(\mu=(\mu_1,\mu_2,\dots,\mu_N)\in\Delta(\Omega)\),
\begin{equation*}
	\phi(\mu) = \argmax_{a\in\mathbb{R}}\sum_{n=1}^N \mu_n u_R(a,\omega_n)
\end{equation*}
is a real number, and \(\phi:\Delta(\Omega)\to\mathbb{R}\) is continuous. Moreover, we assume that \(\phi\) satisfies the following:
\begin{description}
	\item[(Monotonicity)] If \(\mu\) first-order stochastically dominates \(\mu'\), then \(\phi(\mu)>\phi(\mu')\);
	\item[(In-betweenness)] \(\phi\) is strictly quasiconcave and strictly quasiconvex. That is, if \(\phi(\mu)>\phi(\mu')\), then \(\phi(\mu)>\phi(\lambda\mu+(1-\lambda)\mu')>\phi(\mu')\) for all \(\lambda\in(0,1)\). %\gtreqqless
\end{description}
The intuition is that the receiver wants to match the state of the world. He wants to choose a higher action if he believes that higher states are more likely (in the sense of first-order stochastic dominance). An example is when the receiver has quadratic loss utility \(u_R(a,\omega_n)=-(a-\omega_n)^2\). Given any posterior belief, the receiver's unique optimal action equals his expectation of the state, i.e., \(\phi(\mu)=\sum_{n=1}^N\mu_n\omega_n\).

The sender's payoff \(u_S(a,\omega)=a\) depends only on the receiver's action. Given the assumption on the receiver's payoff, the sender has an incentive to persuade the receiver that higher states are more likely. Specifically, the evidence endowment \(e\) or the realized state \(\omega\) is payoff irrelevant to the sender.

\subsection{Strategies}

A (pure) strategy of the sender is \(\mathbf{s}:E\to E\) such that \(\mathbf{s}(e)\precsim e\) for all \(e\in E\). A behavioral strategy of the sender is \(\sigma:E\to\Delta(E)\) such that \(supp(\sigma(\cdot|e))\subset LC(e)\) for all \(e\in E\), where \(LC(e)=\{m:m\precsim e\}\) is the \emph{lower contour set} of \(e\) with respect to the technology of disclosure.

A (pure) strategy of the receiver is \(\mathbf{a}:E\to\mathbb{R}\). A system of beliefs for the receiver is \(\mu:E\to\Delta(\Omega)\), where \(\mu(m)\) denotes the receiver's posterior belief after seeing message \(m\), and \(\mu_n(m)\) denotes the probability of state \(\omega_n\) under such belief. Since the receiver's expected utility maximization problem admits a unique solution given any posterior belief, the receiver only uses pure strategies in equilibrium.

\subsection{Equilibrium}

A disclosure game may have multiple Nash equilibria, including trivial ones. For example, there is a Nash equilibrium of the hedge fund example where the manager always discloses the empty history to the investor, and the investor chooses investment based on his prior belief about the hedge fund. This trivial equilibrium is also a perfect Bayesian equilibrium outcome if the investor uses critical off-path beliefs (i.e., if the manager discloses any return, the investor believes that the actual trading record contains many earlier low returns that are concealed). We focus on perfect Bayesian equilibria which satisfy the following two conditions by \citet{hart_kremer_perry_2017}:
\begin{description}
	\item[(Truth-leaning)] The sender truthfully discloses her evidence endowment \(e\) whenever disclosing \(e\) is optimal;
	\item[(Off-path beliefs)] If an off-path message \(m\) is reported, the receiver believes that the sender is disclosing her evidence endowment without omission, i.e., \(e=m\).
\end{description}

Formally, a \emph{truth-leaning equilibrium} (or simply, an \emph{equilibrium}) is a collection of the sender's strategy, the receiver's strategy, and the receiver's system of beliefs \((\sigma^\star,\mathbf{a}^\star,\mu)\) such that:
\begin{description}
	\item[(Sender optimality)] Given \(\mathbf{a}^\star\),
	\begin{equation*}
		supp(\sigma^\star(\cdot|e))\subset \argmax_{m\precsim e} \mathbf{a}^\star(m)
	\end{equation*}
	for all \(e\in E\);
	\item[(Receiver optimality)] Given \(\mu\),
	\begin{equation*}
		\mathbf{a}^\star=\phi\circ\mu;
	\end{equation*}
	\item[(Bayesian consistency)] Given \(\sigma^\star\), the receiver's posterior belief
	\begin{equation*}
		\mu_n(m)=\frac{\sum_{e\in E}\sigma^\star(m|e)F_n(e)\pi_n}{\sum_{e\in E}\sigma^\star(m|e)F(e)}
	\end{equation*}
	for all \(n\) and all on-path messages \(m\);
	\item[(Truth-leaning)] Given \(\mathbf{a}^\star\),
	\begin{equation*}
		e\in\argmax_{m\precsim e}\mathbf{a}^\star(m) \Rightarrow \sigma^\star(e|e)=1;
	\end{equation*}
	\item[(Off-path beliefs)] The receiver's posterior belief
	\begin{equation*}
		\mu_n(m)=\frac{F_n(m)\pi_n}{F(m)}
	\end{equation*}
	for all \(n\) and all off-path messages \(m\).
\end{description}

Several remarks are in order regarding the refinement. First, the intuition of truth-leaning is that there is a ``slight inherent advantage'' for the sender to tell the truth, and ``there must be good reasons for not telling it.'' \citet{hart_kremer_perry_2017} define truth-leaning equilibrium using an infinitesimally perturbed game. They show that, in a setting with finite evidence spaces, any truth-leaning equilibrium is a Nash equilibria of a perturbed game where (i) there is an infinitesimal probability that the sender's evidence endowment is revealed to the receiver regardless of what she chooses to disclose, and (ii) the sender receives an infinitesimal reward for truthfully disclosing her evidence endowment. Second, a truth-leaning equilibrium is receiver optimal and has the same payoffs as in a mechanism where the receiver commits ex ante to a reward policy.\footnote{It can be easily shown that this equivalence result in \citet{hart_kremer_perry_2017} remains true in our setting where the evidence space is not finite.} As a result, the sender's equilibrium payoff is the same across equilibria, and by the condition on the receiver's off-path beliefs, so is the receiver's equilibrium strategy. Lastly, the condition on off-path beliefs is a tie-breaking rule that applies only when disclosing the evidence endowment truthfully is optimal--it does not require the sender to disclose the optimal message that ``omits the least information'' if reporting the evidence endowment is not optimal.

\subsection{Notation}

We conclude this section by introducing some auxiliary functions and notation related to technologies of disclosure. Let \(\nu\) be a set function on \(E\) taking value in \(\Delta(\Omega)\) such that for all \(n\) and \(A\subset E\) nonempty,
\begin{equation*}
	\nu_n(A)=\frac{F_n(A)\pi_n}{F(A)}.
\end{equation*}
If \(A=\{e\}\) is a singleton, we write it as \(\nu(e)\) (rather than \(\nu(\{e\})\)). It is \emph{the sender's private belief} about the state after observing evidence \(e\), derived from Bayes' rule. By definition, \(\nu(m)\) is also the receiver's posterior belief at an off-path message \(m\) in any truth-leaning equilibrium. Let \(\xi=\phi\circ\nu\) be another set function on \(E\), and similarly, we write \(\xi(e)\) for a singleton set \(\{e\}\). \(\xi(m)\) is the \emph{face value} of a message \(m\). That is, it is the receiver's optimal action if he takes message \(m\) at face value. Given any equilibrium, let \(V:E\to\mathbb{R}\) be such that \(V(e)=\max_{m\precsim e}\mathbf{a}^\star(m)\) for all \(e\in E\). This is referred to as the sender's \emph{value function} and may depend on the choice of equilibrium.

Given a technology of disclosure \(\precsim\), we denote the asymmetric part by \(\prec\) (i.e., \(e\prec e'\Leftrightarrow e\precsim e',e'\not\precsim e\)), and the equivalence relation by \(\sim\) (i.e., \(e\sim e'\Leftrightarrow e\precsim e',e'\precsim e\)). Let \(A\subset E\) and \(L\subset A\) be nonempty. We say \(L\) is a \emph{lower contour set in} \(A\) if for all \(e\in L\) and \(e'\in A\), \(e'\precsim e\Rightarrow e'\in L\).

\section{Equilibrium Characterization} \label{sec3}

\subsection{An illustrative example} \label{sec3.1}

To illustrate our set-up and results, let us solve a simple disclosure game with the so-called \emph{linear evidence}. This can be viewed as a generalization of \citet{wolinsky_2003} to infinite evidence spaces. There are two possible states, a bad state (\(\omega_1=0\)) and a good state (\(\omega_2=1\)), with equal prior (i.e., \(\pi_1=\pi_2=\frac{1}{2}\)). The receiver has quadratic loss payoff function \(u_R(a,\omega)=-(a-\omega)^2\), so given any posterior belief \(\mu=(\mu_1,\mu_2)\in\Delta(\{0,1\})\) his optimal action equals his posterior belief on the good state, i.e., \(\phi(\mu)=\mu_2\). The evidence space \(E=\mathbb{N}\), and the technology of disclosure is the linear order \(\leq\) on \(\mathbb{N}\). That is, the sender with evidence \(e\) can disclose any natural number \(m\leq e\).

The distributions of evidence are given in Table \ref{table1}. Namely, regardless of the state, the sender has evidence \(e=0\) with probability \(\frac{1}{2}\); if \(e\geq 1\), the evidence is distributed geometrically such that in the bad state, the sender has evidence \(e\) with probability \(F_1(e)=\frac{1}{3^e}\), and in the good state, the probability is \(F_2(e)=\frac{1}{2^{e+1}}\). We may interpret evidence \(e=0\) as being uninformative, since after seeing evidence \(e=0\), the sender's private belief about the state is the same as the prior. For all evidence \(e\geq1\), the sender's private belief on the good state \(\nu_2(e)\) is strictly increasing in \(e\). That is, the larger the evidence, the more likely the state is good.

\begin{table}[th]
\centering
\renewcommand{\arraystretch}{1.5}
\begin{tabular}{cccccccc}
	& 0 & 1 & 2 & 3 & $\dots$ & \(e\) & $\dots$ \\ \cline{2-8} 
\multicolumn{1}{c|}{$F_1$} & \multicolumn{1}{c|}{$\frac{1}{2}$} & \multicolumn{1}{c|}{$\frac{1}{3}$} & \multicolumn{1}{c|}{$\frac{1}{9}$} & \multicolumn{1}{c|}{$\frac{1}{27}$} & \multicolumn{1}{c|}{$\dots$} & \multicolumn{1}{c|}{$\frac{1}{3^e}$} & \multicolumn{1}{c}{$\dots$} \\ \cline{2-8}
\multicolumn{1}{c|}{$F_2$} & \multicolumn{1}{c|}{$\frac{1}{2}$} & \multicolumn{1}{c|}{$\frac{1}{4}$} & \multicolumn{1}{c|}{$\frac{1}{8}$} & \multicolumn{1}{c|}{$\frac{1}{16}$} & \multicolumn{1}{c|}{$\dots$} & \multicolumn{1}{c|}{$\frac{1}{2^{e+1}}$} & \multicolumn{1}{c}{$\dots$} \\ \cline{2-8} 
\multicolumn{1}{c|}{$\nu_2$} & \multicolumn{1}{c|}{$\frac{1}{2}$} & \multicolumn{1}{c|}{$\frac{3}{7}$} & \multicolumn{1}{c|}{$\frac{9}{17}$} & \multicolumn{1}{c|}{$\frac{27}{43}$} & \multicolumn{1}{c|}{$\dots$} & \multicolumn{1}{c|}{$\frac{3^e}{3^e+2^{e+1}}$} & \multicolumn{1}{c}{$\dots$} \\ \cline{2-8}
\end{tabular}
\caption{Distributions of evidence} \label{table1}
\end{table}

In any equilibrium, sender type \(e=1\) discloses message \(m=0\). Suppose that this is not the case. That is, there exists an equilibrium \((\sigma^\star,\mathbf{a}^\star,\mu)\) with \(\sigma^\star(1|1)>0\). By sender optimality, the receiver's action after seeing message 1 is at least his action after seeing message 0. Since the receiver's action equals his posterior belief on the good state, \(\mu_2(1)\geq\mu_2(0)\). However, \(\nu_2(1)<\nu_2(0)\), so there must be sender types \(e>1\) that disclose message 1 with positive probability. Let \(Z=\{e>1:\sigma^\star(1|e)>0\}\) be the set of such sender types. By truth-leaning and sender optimality, every piece of evidence \(e\in Z\) is off-path as a message, so \(\mu_2(1)>\mu_2(e)=\nu_2(e)\). Suppose that \(Z\) is finite, and let \(e^\star\) be the largest member in \(Z\). For all sender types \(e\in Z\), \(\nu_2(e)\leq\nu_2(e^\star)\), and for sender type \(e=1\), \(\nu_2(1)<\nu_2(e^\star)\). Therefore, after observing message 1, the receiver's posterior belief \(\mu_2(1)<\nu_2(e^\star)\). This shows that the sender type \(e^\star\) can profitably deviate to disclosing \(e^\star\) truthfully and thereby receiving \(\xi(e^\star)=\nu_2(e^\star)\), instead of disclosing message 1 and receiving \(\mathbf{a}^\star(1)=\mu_2(1)\). Hence, \(Z\) is infinite, and \(\mu_2(1)\geq\sup_{e\in Z}\nu_2(e)=1\). But this is not possible. Therefore, in any equilibrium, the sender discloses message 0 if \(e=1\), and message 1 is an off-path message.

We now argue that all sender types \(e>1\) disclose truthfully. Suppose that there exist an equilibrium \((\sigma^\star,\mathbf{a}^\star,\mu)\) and \(\hat{e}>\hat{m}\), such that \(\hat{e}>1\) and \(\sigma^\star(\hat{m}|\hat{e})>0\). By truth-leaning and sender optimality, \(\hat{e}\) is off-path as a message, and \(\mu_2(\hat{m})>\mu_2(\hat{e})=\nu_2(\hat{e})\). However, \(\nu_2(e)<\nu_2(\hat{e})\) for all \(e<\hat{e}\). Therefore, there must be some sender type \(e>\hat{e}\) such that \(\sigma^\star(\hat{m}|e)>0\). Since \(\hat{e}\) is arbitrary, this shows that there are infinitely many sender types that disclose \(m^\star\) with positive probability. Sender optimality thus requires \(\mu_2(\hat{m})=1\), which is not possible.

The above shows that the sender discloses truthfully if \(e\neq1\) and discloses message 0 if \(e=1\). Receiver optimality and the condition on off-path beliefs pin down the receiver's system of beliefs and strategy: after seeing message \(m=0\), \(\mathbf{a}^\star(0)=\mu_2(0)=\nu_2(\{0,1\})=\frac{9}{19}\); after seeing message \(m\geq 1\), \(\mathbf{a}^\star(m)=\mu_2(m)=\nu_2(m)=\frac{3^m}{3^m+2^{m+1}}\). As can be easily checked, this is the unique equilibrium of the game.

Another way to think of the equilibrium is to think of the partition \(\{\{0,1\},\{2\},\{3\},\dots\}\) of the evidence space. Sender types in the set \(\{0,1\}\) pool on disclosing message 0, and other sender types disclose truthfully. Bayes' rule requires that the receiver's on-path belief \(\mu(0)\) be equal to \(\nu(\{0,1\})\), and the off-path belief be \(\mu(1)=\nu(1)\). Lastly, receiver optimality pins down the receiver's strategy.

\subsection{Main results} \label{sec3.2}

For general disclosure games, Theorem \ref{thm1} characterizes the unique equilibrium value function of the sender. Since the receiver's action at an on-path message equals the sender's equilibrium payoff, and that at an off-path message is determined by the refinement, the receiver's strategy is the same in all equilibria. If the state space is binary, the receiver's system of beliefs is also the same, and all equilibria differ only in the sender's behavioral strategies.

\begin{thm} \label{thm1}
	A function \(V:E\to\mathbb{R}\) is the sender's value function in an equilibrium if and only if it satisfies the following conditions:
	\begin{enumerate}
		\item \(V:(E,\precsim)\to(\mathbb{R},\leq)\) is non-decreasing;
		\item \(\xi(V^{-1}(x))=x\) for all \(x\in Range(V)\);
		\item \(\xi(L)\geq x\) for all \(x\in Range(V)\) and all lower contour sets \(L\) in \(V^{-1}(x)\).
	\end{enumerate}
	Moreover, the function satisfying the above conditions is unique.
\end{thm}

Notice that \(V^{-1}\) is a correspondence from \(Range(V)\) to \(E\), and \(\xi\) is a set function on \(E\). The ``only if'' part is straightforward. Condition 1 is due to sender optimality. In the example in Section \ref{sec3.1}, if the equilibrium payoff of the sender type \(e=2\) is less than that of \(e=0\), then disclosing 0 is a profitable deviation for the sender type \(e=2\). Condition 2 is by Bayesian consistency. In the example, message \(m=0\) is disclosed by the sender if \(e=0\) or \(e=1\). Bayes' rule therefore requires that \(\mu(0)=\nu(\{0,1\})\). Therefore, \(V(0)=\mathbf{a}^\star(0)=\phi(\mu(0))=\xi(\{0,1\})\). Notice that conditions 1 and 2 hold in any perfect Bayesian equilibrium. Condition 3 is true in truth-leaning equilibria. Suppose that, in the example, \(\xi(0)<\xi(\{0,1\})\). Then \(\xi(1)>\xi(\{0,1\})\). Since \(m=1\) is an off-path message, \(\xi(1)\) is the receiver's action after observing message \(m=1\). Therefore, disclosing truthfully is a profitable deviation for the sender type \(e=1\).

The ``if'' part is more involved. Our proof uses a straightforward induction process which applies to disclosure games with infinite evidence spaces. This process provides a way to construct equilibrium strategies from the sender's equilibrium value function.

The induction process is outlined here. Given a function \(V\) satisfying the conditions in Theorem \ref{thm1}, we want to find an equilibrium in which \(V\) is the sender's value function. First, we observe that in any equilibrium, if \(V(e)>\xi(e)\) for some sender type \(e\in E\), then \(e\) as a message is off-path, i.e., \(\mu(e)=\nu(e)\), and \(\mathbf{a}^\star(e)=\xi(e)\); if \(V(e)\leq\xi(e)\), then the sender discloses truthfully when her evidence endowment is \(e\), and \(\mathbf{a}^\star(e)=V(e)\). Since \(\xi\) is determined by the distributions of evidence \(F_1,F_2,\dots,F_N\) and \(\phi\), given any function \(V\), the set of on-path messages, the receiver's strategy, and the receiver's off-path beliefs are determined in any candidate equilibrium. The remaining task is to construct a strategy of the sender and the receiver's beliefs at on-path messages, such that the strategy is sender optimal and truth-leaning, and the beliefs are consistent with Bayes' rule. To do so, let us consider the following induction process. Suppose that we have a candidate strategy of the sender defined on a lower contour set \(L\) in \(E\) with the following properties: (i) sender optimality and truth-leaning are satisfied on \(L\), and (ii) given the sender's strategy, the receiver's sequentially rational action after observing an on-path message \(m\) in \(L\) is at least \(V(m)\). If \(L=E\), the candidate strategy is an equilibrium strategy, and \(V\) is the sender's value function in the resulted equilibrium. If \(L\subsetneq E\), condition 3 in Theorem \ref{thm1} ensures that there are slacks in property (ii), and we can extend the candidate strategy of the sender to a superset of \(L\). That is, we can find another candidate strategy which is defined on a lower contour set \(\tilde{L}\supsetneq L\) in \(E\) such that properties (i) and (ii) are satisfied on \(\tilde{L}\), and the two candidate strategies coincide on \(L\). If the evidence space is finite, this extension process solves an equilibrium strategy. If the evidence space is not finite, using transfinite induction, we show existence of an equilibrium in which \(V\) is the sender's value function.

Since the sender's value function is the same in all equilibria, we shall use the term without referring to a specific equilibrium. As is in Section \ref{sec3.1}, an equilibrium can be considered as a partition of the evidence space. The following corollary formalizes this idea using the notion of ordered partition\footnote{An \emph{ordered partition} \citep{stanley_1999} is a partition with a total order defined on the elements of the partition. Recall that a \emph{total order} \(\leq\) (on a set \(X\)) is a partial order that is also \emph{connex} (\(\forall x,x'\in X\), \(x\leq x'\) or \(x'\leq x\)).} and states an equivalent result of Theorem \ref{thm1}.

\begin{crl} \label{crl2}
	An equilibrium exists if and only if there exists an ordered partition of \(E\), denoted by \((\mathcal{A},\leq_P)\), satisfying:
	\begin{enumerate}
		\item \(\xi:(\mathcal{A},\leq_P)\to([0,1],\leq)\) is strictly increasing;
		\item For all \(A_1,A_2\in\mathcal{A}\) and \(e_1\in A_1,e_2\in A_2\), \(e_1\precsim e_2\Rightarrow A_1\leq_P A_2\);
		\item For all \(A\in\mathcal{A}\) and all lower contour sets \(L\) in \(A\), \(\xi(L)\geq\xi(A)\).
	\end{enumerate}
	The sender's value function \(V\) satisfies \(V(e)=\xi(\mathcal{A}(e))\) for all \(e\in E\), where \(\mathcal{A}(e)\in\mathcal{A}\) is such that \(\mathcal{A}(e)\ni e\). Moreover,the ordered partition satisfying the above conditions is unique.
\end{crl}

We refer to the ordered partition \((\mathcal{A},\leq_P)\) in Corollary \ref{crl2} as the \emph{equilibrium partition}. The characterization of the equilibrium partition is similar to the characterization of the optimal mechanism outcome by \citet{rappoport_2017}. In the optimal mechanism, the receiver commits to choosing action \(\xi(A)\) after seeing any message in a set \(A\) in the partition. The first two conditions in Corollary \ref{crl2} ensure that truthfully revealing evidence is incentive compatible for the sender, and the last condition ensures that the mechanism is receiver optimal (i.e., there does not exist a finer partition that satisfies incentive compatibility). The similarity between our characterizations is not surprising, since the sender's ex post payoffs are the same in truth-leaning equilibria and the optimal mechanism. However, the result by \citet{rappoport_2017} only applies to finite evidence spaces and cannot be used to solve truth-leaning equilibrium strategies. With the aforementioned induction process, we are able to construct truth-leaning equilibria, even when the evidence space is infinite.

It is also clear from Corollary \ref{crl2} that the equilibrium partition is independent of the receiver's payoff function, as long as the solution to the receiver's maximization problem satisfies monotonicity and in-betweenness. That is, fixing an evidence structure and a technology of disclosure, the equilibrium partition is the same across all disclosure games we consider. If \(\phi\) is weakly monotonic and satisfies weak in-betweenness (e.g., when the receiver's action space is finite), the truth-leaning equilibrium partition still satisfies the characterization in Corollary \ref{crl2}. However, a truth-leaning equilibrium may fail to exist even if there exists an (unique) ordered partition satisfying the characterization. \citet{jiang_2020} studies finite action disclosure games in more details. In such games, the ordered partition given by Corollary \ref{crl2} characterizes purifiable equilibrium, which is an amended solution concept to truth-leaning equilibrium in finite disclosure games.

Theorem \ref{thm3} below further characterizes the equilibrium partition.

\begin{thm} \label{thm3}
	An ordered partition \((\mathcal{A},\leq_P)\) of \(E\) is the equilibrium partition if and only if every \(A\in\mathcal{A}\) is the largest lower contour set in \(E\setminus\bigcup_{A'<_P A}A'\) that minimizes \(\xi\). That is, \(A\) solves
	\begin{equation} \label{eq2.7}
		\min_{L}\xi(L) \quad s.t. \quad L~\text{is a lower contour set in}~E\setminus\bigcup_{A'<_P A}A',
	\end{equation}
	and for all \(L\) that solves (\ref{eq2.7}), \(L\subset A\).
\end{thm}

Following Theorem \ref{thm3}, we can solve the equilibrium partition by recursively solving the largest lower contour set that minimizes \(\xi\). The induction process is as follows. First, solve the minimization problem
\begin{equation*}
	\min_{L}\xi(L) \quad s.t. \quad L~\text{is a lower contour set in}~E.
\end{equation*}
Lemma \ref{lmm.a5} in Appendix \ref{sec.a5} shows that if the above problem admits a solution, it has a largest minimizer. Denote by \(A_0\) this largest minimizer. Now suppose that for some \(n\in\mathbb{N}\), \(A_0,A_1,\dots,A_n\) are such that every \(A_k\) is the largest lower contour set in \(E\setminus\bigcup_{i<k}A_i\) that minimizes \(\xi\). We solve the minimization problem
\begin{equation*}
	\min_{L}\xi(L) \quad s.t. \quad L~\text{is a lower contour set in}~E\setminus\bigcup_{i\leq n} A_i.
\end{equation*}
If the above minimization problem admits a minimizer, let \(A_{n+1}\) be its largest minimizer, and we can proceed by induction.

By Theorem \ref{thm3}, if the equilibrium partition is finite, then the above induction process ends in finite steps, and the resulted partition \(\mathcal{A}=\{A_0,A_1,\dots,A_n\}\), along with the total order \(\leq_P\) such that \(A_i \leq_P A_j\Leftrightarrow i\leq j\), gives the equilibrium partition. Specifically, the above induction process solves the equilibrium partition when the evidence space is finite. If the equilibrium partition is not finite, but sets in the equilibrium partition \(\mathcal{A}\) can be indexed by \(\mathbb{N}\) such that \(A_i \leq_P A_j\Leftrightarrow i\leq j\), then every \(A_n\) is solved by the above induction process. Conversely, if every \(A_n\) in the above induction process is well defined, and \(\bigcup_{n\in\mathbb{N}}A_n=E\), then an equilibrium exists, and the equilibrium partition is given by \(\mathcal{A}=\{A_n\}_{n\in\mathbb{N}}\) and \(\leq_P\) such that \(A_i \leq_P A_j\Leftrightarrow i\leq j\). 

\citet{rappoport_2017} suggests an alternative method to compute the equilibrium partition by recursively solving the largest upper contour set that maximizes \(\xi\). A recent version of \citet{rappoport_2017} notes that the two methods yield the same outcome when the evidence space is finite. However, in many disclosure games with large evidence spaces, there does not exist an upper contour set that maximizes \(\xi\), so the alternative methods does not apply. For example, in the example in Section \ref{sec3.1}, an upper contour set in \(E\) has the form of \(\{e:e\geq n\}\) where \(n\in\mathbb{N}\). Since the larger the evidence, the more likely the state is good, \(\xi(\{e:e\geq n\})\) is strictly increasing in \(n\) for \(n\geq 1\) (as \(n\to\infty\), \(\xi(\{e:e\geq n\})\to 1\)). Hence, there does not exist an upper contour set that maximizes \(\xi\).

In general, when the evidence space is infinite, whether or not we can solve the equilibrium partition by recursively solving the largest lower contour set that minimizes \(\xi\) depends on the distributions of evidence \(F_1,F_2,\dots,F_N\), as well as the technology of disclosure \(\precsim\). The following theorem provides a sufficient condition for equilibrium existence which depends only on the technology of disclosure. Moreover, if this condition is satisfied, the equilibrium partition can be solved using the induction process above, as is shown in Appendix \ref{sec.a5}. We say a set \(A\subset E\) is \emph{connected} if for all \(e,e'\in A\), there exists a finite sequence \(e_0,e_1,\dots,e_K\) in \(A\) such that \(e_0=e\), \(e_N=e'\), and either \(e_k\precsim e_{k+1}\) or \(e_k\succsim e_{k+1}\) for all \(k=0,1,\dots,K-1\).

\begin{thm}[Existence] \label{thm4}
	An equilibrium exists if every connected set in \(E\) has finitely many minimal elements.
\end{thm}

The sufficient condition in Theorem \ref{thm4} is quite general and is satisfied by many evidence structures. Here are some examples with countable evidence spaces:

\begin{enumerate}
	\item Linear evidence \((\mathbb{N},\leq)\) in Section \ref{sec3.1}.
	\item Multidimensional evidence \((\mathbb{N}^d,\leq)\). The sender is endowed with \(d\) types of evidence, so her evidence is a vector \((e_1,e_2,\dots,e_d)\in\mathbb{N}^d\), where \(e_i\) denotes the amount of type \(i\) evidence she has. The sender can omit evidence of any type. Hence, \((m_1,m_2,\dots,m_d)\leq (e_1,e_2,\dots,e_d)\) if \(m_i\leq e_i\) for all \(i\).
	\item Lexicographical evidence \((\mathbb{N}^2,\leq^{lex})\). The sender is endowed with primary and secondary evidence, so her evidence endowment is a vector \((p,s)\in\mathbb{N}^2\), where \(p\) and \(s\) are the amounts of primary and secondary evidence she has, respectively. The sender can disclose all primary evidence and a part of her secondary evidence. Alternatively, she can disclose less primary evidence (e.g., she can make a weaker argument). In this case, secondary evidence is not material, and she can claim to have any amount of secondary evidence. Hence, \((m_p,m_s)\leq^{lex}(p,s)\) if \(m_p=p\) and \(m_s\leq s\), or if \(m_p<p\).
	\item \citet{dye_1985} evidence structure \((E,\precsim)\). The evidence space \(E=\{e_0,e_1,\dots\}\) is a finite or countable set. The sender is either informed with some evidence \(e_n\) (\(n>1\)), or is uninformed (in which case we denote her evidence by \(e_0\)). If the sender is informed, she can either truthfully disclose or pretend to be uninformed; if the sender is uninformed, she can only disclose \(e_0\). That is, \(m\precsim e\) if \(m=e\) or \(m=e_0\).
	\item Sequential evidence. The sender's evidence is a finite sequence of signals \((s_1,s_2,\dots,s_T)\in\cup_{t=1}^\infty S^t\), where \(S\) is a finite set of signals. The sender can truncate the sequence and disclose \((s_1,s_2,\dots,s_k)\) for some \(k\leq T\). Section \ref{sec4} studies a model with a similar evidence structure.
\end{enumerate}

One evidence structure that violates the sufficient condition in Theorem \ref{thm4} is as follows. Let \((\mathbb{N},\leq)\) be the standard linear evidence in Section \ref{sec3.1}, but now in addition to observing an evidence \(e\in\mathbb{N}\), the sender has a private type, which can be either strategic (\(s\)) or honest (\(h\)). The type is independent of the state or the evidence endowment. The sender can disclose any \(m\leq e\) if she is strategic, but she can only disclose truthfully if she is honest, i.e., \(m=e\). We can model this game as a disclosure game where the evidence space \(E=\mathbb{N}\times\{s,h\}\), and the disclosure technology is as follows: \((e',t')\precsim(e,t)\) if \(e'=e\) and \(t'=t=h\), or if \(e'\leq e\) and \(t=s\).\footnote{Although the sender does not disclose her private type, we can study this game using the disclosure game framework, since disclosing \(t'=h\) is weakly dominant.} However, the evidence structure \((E,\precsim)\) does not satisfy the condition in Theorem \ref{thm4}. For example, \(E\) itself is a connected set, but the set of minimal elements in \(E\) is \(\mathbb{N}\times\{h\}\), which is not finite.\footnote{Nevertheless, if the evidence is distributed according to Table \ref{table1} and independent of the sender's type, it is easy to verify that an equilibrium exists, in which the sender discloses \((0,h)\) if \((e,t)=(0,s)\) or \((e,t)=(1,s)\), and discloses truthfully otherwise.}

% More examples where the condition is violated
% \begin{enumerate} [resume]
% 	\item Integer factorization. The sender's evidence is \(e\in E\{2,3,4,\dots\}\), and she can disclose any integer \(m\in E\) that devides \(e\). That is, \(m\precsim e\) if \(\frac{e}{m}\in\mathbb{Z}\). \(E\) is connected, but the set of minimal elements in \(E\) is the set of all prime numbers, which is infinite.
% 	\item Rounding. The sender's evidence is \(e\in E=\{\frac{k}{2}:k\in\mathbb{N}\}\), and she can report the closest integer if \(e\notin\mathbb{N}\). That is, \(m\precsim e\) if \(m\in\{e,\lfloor e\rfloor,\lceil e\rceil\}\). \(E\) is connected, but the set of minimal elements in \(E\) is \(\mathbb{N}\), which is infinite.
% \end{enumerate}

\section{Disclosure of Trading Record} \label{sec4}

We model the interaction between the hedge fund manager and the investor as a disclosure game with left censoring technology of disclosure. The hedge fund manager is the sender, and the investor is the receiver. The state of the world is binary and corresponds to the hedge fund's quality. With probability \(\pi_1\in(0,1)\), the state is bad (\(\omega_1=0\)); with probability \(\pi_2=1-\pi_1\), the state is good (\(\omega_2=1\)).

\subsection{Evidence}

The evidence space \(E=\bigcup_{t=0}^\infty\{g,b\}^t\). A piece of evidence \(e\in E\) is a sequence of binary signals. A typical element of length \(L>0\) is denoted by \(e=(s_1,s_2,\dots,s_L)\), where each \(s_t\in\{g,b\}\). This represents a trading record of the incubation trial, and each signal \(s_t\) represents the return in the \(t\)-th week in the incubation period--\(s_t=g\) stands for high return, and \(s_t=b\) stands for low return.

There is a ``null'' evidence (i.e., the evidence of length zero), denoted \(\emptyset\). For every \(e\in E\), let \(L(e)\) denote the length of \(e\), \(G(e)\) the number of \(g\)'s in \(e\), \(B(e)\) the number of \(b\)'s in \(e\), and \(D(e)=G(h)-B(h)\) the difference between the number of \(g\)'s and \(b\)'s. Moreover, for every \(e\) of length \(L>0\) and \(0<k\leq L\), let \(e|_k=(s_{L-k+1},s_{L-k+2},\dots,s_L)\) denote a \emph{left truncation} of \(e\). As a convention, \(e|_0=\emptyset\) for all \(e\in E\).

Let the distribution of evidence under the bad state \(F_1\) and that under the good state \(F_2\) be as follows:
\begin{align*}
	F_1(e) &= (1-p-q)p^{B(e)}q^{G(e)}, \\
	F_2(e) &= (1-p-q)p^{G(e)}q^{B(e)},
\end{align*}
where \(p,q\) are parameters satisfying \(1>p>q>0\) and \(1>p+q>0\). An implication of the above specifications is that the length of evidence is distributed independent of the realized state. Given the length of evidence, each signal is an independent draw from a Bernoulli distribution. With probability \(\frac{p}{p+q}\), the signal indicates the realized state; with the complement probability \(\frac{q}{p+q}\), it indicates the alternative state. The informativeness of each signal is constant over time and symmetric across states.

\subsection{Left censoring technology of disclosure}

Given her evidence endowment, the sender can disclose a set of most recent signals and conceal earlier signals, i.e., she can report a truncation of her evidence endowment. This defines the \emph{left censoring technology of disclosure}. Formally, it is a partial order \(\precsim\) on \(E\) such that \(e_1\precsim e_2\Leftrightarrow e_1=e_2|_{L(e_1)}\).

\subsection{Payoffs}

The receiver has quadratic loss utility function \(u_R(a,\omega)=-(a-\omega)^2\). Given his posterior belief \(\mu=(\mu_1,\mu_2)\in\Delta(\{0,1\})\), the receiver's optimal action equals his posterior belief on the good state, i.e., \(\phi(\mu)=\mu_2\). The sender's payoff \(u_S(a,\omega)=a\).

\subsection{Equilibrium features}

Applying results from Section \ref{sec3}, we can solve the sender's value function.

\begin{pro}[Equilibrium value function]\label{pro4}
	Let \(M:E\to\mathbb{N}\) be defined by
	\begin{equation*}
		M(e) = \max_{k\leq L(e)} D(e|_k).
	\end{equation*}
	An equilibrium exists, and the sender's value function is
	\begin{equation} \label{eq2.9}
		V(e)=\frac{1}{1+\frac{\pi_1}{\pi_2}\frac{p-\alpha}{q-\alpha}\left(\frac{q}{p}\right)^{M(e)+1}},
	\end{equation}
	where \(\alpha=\frac{1-\sqrt{1-4pq}}{2}\).
\end{pro}

The proposition states that in any equilibrium, given trading record \(e\), the sender's equilibrium value is increasing in \(M(e)\), which is the maximal difference between the number of \(g\)'s and the number of \(b\)'s among all left truncations of \(e\). Sender types with different evidence endowment but the same value of maximal difference receive the same payoff in any equilibrium. For example, if the sender's evidence is \((g,b)\) or \((b,g,b)\), her equilibrium payoff is the same as having the null evidence \(\emptyset\).

When \(q\to0\), any signal is fully revealing of the state. The sender receives the null evidence with probability \(1-p\), receives only good signals with probability \(\pi_2 p\), and receives only bad signals with probability \(\pi_1p\). In the unique equilibrium, the sender discloses truthfully if she has any good signal and discloses the null evidence otherwise. In the first case, the receiver knows that the state is good and chooses action 1. This corresponds to the fact that \(V(e)\to 1\) as \(q\to0\) if \(M(e)>0\). In the latter case, the receiver's posterior belief on the good state after seeing the null evidence is \(\frac{\pi_2(1-p)}{1-\pi_2 p}\). This is also the limit of the sender's payoff \(V(e)\) as \(q\to0\) if \(M(e)=0\).

We now characterize the sender's equilibrium strategy. A seemingly optimal strategy for the sender with evidence endowment \(e\) is to report a maximal difference truncation \(e|_k\) such that \(D(e|_k)=M(e)\). Disclosing any other message \(m\) appears to be sub-optimal, since \(m\) is worse than \(e|_k\) if interpreted at face value (i.e., \(\xi(m)<\xi(e|_k)\)). However, this is not necessarily the case. In fact, the sender sometimes reports longer messages which include redundant and seemingly unfavorable information at the beginning of the messages. The receiver, on the other hand, is not influenced by the unfavorable information in on-path messages and acts the same as if only \(e|_k\) is reported.

Notice that for all \(e\in E\), \(\xi(e)\geq V(e)\) if and only if \(D(e)\geq M(e)-n^\star\), where
\begin{equation} \label{eq2.10}
	n^\star=\log_{p/q}\left(\frac{p-\alpha}{q-\alpha}\right)-1>0
\end{equation}
is a constant determined solely by the model parameters \(p\) and \(q\). Hence, in any equilibrium, evidence \(e\) is disclosed truthfully so long as \(D(e)\geq M(e)-n^\star\). This demonstrates our first point. That is, seemingly unfavorable information may be reported on-path. Specifically, if
\begin{equation} \label{eq2.11}
	(p+q)(1-p-q)\leq pq,
\end{equation}
then \(n^\star\geq1\), and in any equilibrium, some seemingly unfavorable information is disclosed. For instance, consider evidence \(e=(b)\) and observe that \(D((b))=-1\) and \(M((b))=0\). Therefore, if the parametric assumption (\ref{eq2.11}) is satisfied, evidence \((b)\) is disclosed truthfully in any equilibrium of the game, even though reporting the null evidence is feasible and seemingly better, as in terms of face value, \(\xi(\emptyset)>\xi((b))\). Similarly, for all evidence \(e\) such that \(e=(b,e')\), where \(e'\) is a maximal difference truncation of \(e\), the sender truthfully discloses \(e\) in equilibrium if (\ref{eq2.11}) is satisfied, even though \(e\) contains an initial bad signal. This result may appear similar to the result in \citet{dziuda_2011}, but the strategic reasoning underlying these results are different. In \citet{dziuda_2011}, the sender discloses unfavorable signals in order to be perceived as being honest. In the current paper, if all sender types whose evidence endowment has a same maximal difference truncation pool on disclosing that maximal difference truncation, the persuasiveness of that truncation is so undermined that some sender types will find it profitable to disclose truthfully. Thus instead, the sender utilizes several messages in equilibrium, including those seemingly sub-optimal ones, such that all on-path messages have the same persuasiveness and yield a strictly higher payoff than deviating to other feasible messages.

Secondly, when there are multiple maximal difference truncations, the sender's message always contains the longest maximal difference truncation. That is, in any equilibrium, \(L(m)\geq\max(\argmax_k D(e|_k))\) for all \(e\in E\) and \(m\in supp(\sigma^\star(\cdot|e))\). For instance, the sender with the trading record \((b,b,b,g,b)\) will always include the two most recent signals \((g,b)\) in her message, although both truncations \((g,b)\) and \(\emptyset\) are maximal difference truncations. This result follows from the fact that \(\xi(E_0^-)=\xi(E_0)\), where \(E_0^-\) is the set containing the null evidence and all evidence such that every nonempty truncation has strictly more \(b\)'s than \(g\)'s, and \(E_0\) is the set of evidence such that \(M(e)=0\). Notice that \(E_0^-\) is a lower contour set in \(E_0\). Therefore, no sender type \(e\in E_0\setminus E_0^-\) can report a message \(m\in E_0^-\) with positive probability in any equilibrium: after accommodating all reports from sender types in \(E_0^-\), the receiver's sequentially rational action after observing an on-path message in \(E_0^-\) is already \(\xi(E_0)\); any further pooling from sender types outside the set \(E_0^-\) will make the receiver's action too low, and disclosing truthfully will be a profitable deviation for some sender types. Specifically, since \((b,b,b,g,b)\in E_0\setminus E_0^-\), the sender with evidence endowment \((b,b,b,g,b)\) will disclose at least the two most recent signals \((g,b)\) in her message.

Lastly, since all evidence \(e\) with the same value of \(M(e)\) yield the same sender's equilibrium payoff, the receiver's action is not affected by the unfavorable information revealed by the sender in equilibrium. Unless the message contains so many bad signals that it is not an on-path message, the receiver will base his action solely on the favorable information in the message, i.e., his posterior depends on \(M(m)\) instead of \(D(m)\).

The proposition below summarizes the results.

\begin{pro}[Equilibrium strategies] \label{pro5}
	In any equilibrium \((\sigma^\star,\mathbf{a}^\star,\mu)\), for all \(e\in E\):
	\begin{enumerate}
		\item \(\mathbf{a}^\star(e)=\mu_2(e)=\min\{V(e),\xi(e)\}\), where \(V\) is defined in (\ref{eq2.9});
		\item \(\sigma^\star(e|e)=1\) if and only if \(D(e)\geq M(e)-n^\star\), where \(n^\star>0\) is defined in (\ref{eq2.10});
		\item If \(\sigma^\star(m|e)>0\), then \(L(m)\geq\max(\argmax_k D(e|_k))\).
	\end{enumerate}
\end{pro}

\subsection{Comparative statics}

The following proposition shows how the sender's equilibrium value changes with respect to parameters \(p\) and \(q\). Notice that increasing, say \(p\), while holding \(q\) constant introduces two effects: on the one hand, it increases the informativeness of each signal; on the other hand, the expected length of the sender's evidence endowment increases. We introduce the following notation to disentangle the two effects. Let \(\gamma=\frac{p}{q}>1\) be the informativeness of each signal, and \(\kappa=\frac{1}{1-p-q}\) the expected length of the sender's evidence endowment. Given \(\gamma>1\) and \(\kappa>1\), \(p\) and \(q\) are determined.

\begin{pro}[Comparative statics] \label{pro7}
	The following are true:
	\begin{enumerate}
		\item \(V(e)\) is decreasing in \(\kappa\) for all \(e\in E\).
		\item If \(M(e)=0\), \(V(e)\) is decreasing in \(\gamma\). If \(M(e)>0\), there exists \(\hat{\gamma}>1\) such that \(V(e)\) is increasing in \(\gamma\) when \(\gamma>\hat{\gamma}\), and the threshold \(\hat{\gamma}\to 1\) as \(M(e)\to\infty\).
		\item Fixing \(e\in E\) and \(q\in(0,\frac{1}{2})\), there exists \(\hat{p}\in[q,1-q)\) such that \(V(e)\) is increasing in \(p\) if \(p<\hat{p}\) and decreasing in \(p\) if \(p>\hat{p}\), and the threshold \(\hat{p}\) is non-decreasing in \(M(e)\).
	\end{enumerate}
\end{pro}

\begin{figure}[t]
\centering
\begin{subfigure}{.48\textwidth}
\subcaption{\(V(e)\) as a function of \(\kappa\) and \(\gamma\) (\(N(e)=2\))}
\includegraphics[width=\textwidth,trim={32pt 0 32pt 0},clip]{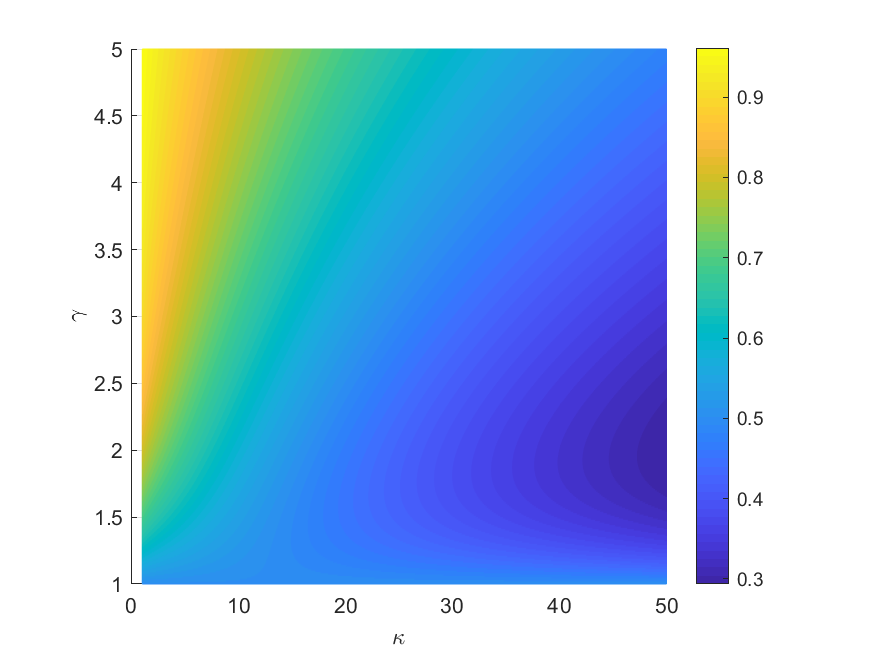}
\end{subfigure}
\begin{subfigure}{.48\textwidth}
\subcaption{The sender's values as functions of \(p\) (\(q=\frac{1}{5}\))}
\includegraphics[width=\textwidth,trim={32pt 0 32pt 0},clip]{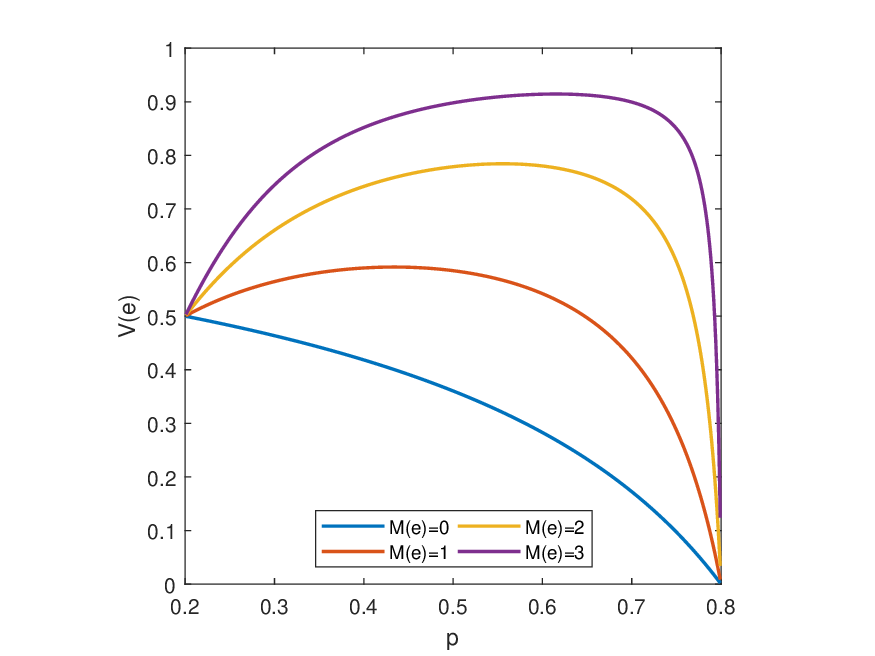}
\end{subfigure}
\caption{Comparative statics results (\(\pi_0=\frac{1}{2}\))}\label{fig3}
\end{figure}

The proposition shows different effects of increasing the informativeness of signals and increasing the length of evidence on the sender's equilibrium value. As is demonstrated in Figure \ref{fig3}(A) for \(M(e)=2\), increasing the expected length of the sender's evidence endowment always decreases the sender's equilibrium value, whereas the sender may benefit from increased informativeness of signals. When \(\kappa\) is small, the sender's equilibrium value is monotone increasing in \(\gamma\); when \(\kappa\) is large, the sender's equilibrium value first decreases and then increases as \(\gamma\) increases. The proposition shows that there exists a threshold \(\hat{\gamma}>1\), independent of \(\kappa\), such that the sender's value increases in \(\gamma\) for all \(\gamma>\hat{\gamma}\) and \(\kappa>1\). An equivalent interpretation of the result is that the effect of increasing the informativeness of signals depends on the sender's evidence endowment. Fixing \(\gamma>1\), there is a threshold \(\hat{n}\geq1\) such that the sender benefits from increasing the informativeness of signals if her evidence endowment is such that \(M(e)>\hat{n}\). That is, increasing the informativeness of signals benefits the sender if her evidence endowment is relatively favorable.

Combining the two competing effects, the effect of increasing \(p\) on the sender's equilibrium value depends on the sender's evidence endowment and is generally not monotonic. This is demonstrated in Figure \ref{fig3}(B). If the sender's evidence endowment is relatively unfavorable (specifically, if \(M(e)=0\)), increasing \(p\) monotonically decreases the sender's equilibrium payoff. If the sender's evidence endowment is relatively favorable, the sender's equilibrium value first increases and then decreases as \(p\) increases. When \(p\to1-q\), the expected length of the sender's evidence endowment tends to infinity. Therefore, regardless of what the sender discloses, the receiver believes that there are infinitely many signals being omitted, and his skepticism becomes so overwhelming that the sender receives close to zero payoff in equilibrium.

%From old discussion section:
%The reason that an equilibrium may fail to exist is as follows. Given a function \(V:E\to\mathbb{R}\) that satisfies the conditions in Theorem \ref{thm5} and \(x\in Range(V)\), we can distinguish three types of evidence in \(V^{-1}(x)\). First, \(E_1=\{e:\xi(e)<V(e)=x\}\) is the set of off-path messages. The sender with evidence \(e\in E_1\) will disclose some on-path messages that leads to the receiver action \(x\). On-path messages belong to one of the two sets, \(E_2=\{e:\xi(e)=V(e)=x\}\) and \(E_3=\{e:\xi(e)>V(e)=x\}\). For a message \(m\in E_3\), \(\phi(\nu(m))>x\), but receiver optimality requires that \(\phi(\mu(m))=x\). Therefore, there must be senders with less favorable evidence disclosing \(m\) with positive probability such that the posterior belief \(\mu(m)\) is low enough. However, due to truth-leaning, only senders with evidence in \(E_1\) can disclose \(m\) with positive probability, and senders with evidence in \(E_2\) must disclose truthfully. This leads to the possibility that for some \(m\in E_3\), even if all senders with evidence in \(E_1\) disclose \(m\), \(\phi(\mu(m))\) still exceeds \(x\), hence violating receiver optimlity. If this happens, there does not exist a truth-leaning equilibrium.

\section{Conclusion}\label{sec6}

There are many situations where communication relies on hard evidence. However, even when information is verifiable, an agent may withhold key information from others if their interests are not perfectly aligned.

We study a general framework of disclosure games. This general framework is compatible with large evidence spaces and a wide range of technologies of disclosure. We characterize the unique equilibrium value function of the sender and present two induction processes. The first one constructs an equilibrium from the sender's equilibrium value function, and the second process solves the sender's equilibrium value function for a broad class of disclosure games.

We apply these results to study the interaction between a hedge fund manager and a potential investor as a disclosure game with large evidence spaces and left-censoring technology. The state space is binary. The sender's evidence is a sequence of binary signals, and the sender can disclose a left truncation of her evidence to the receiver. In equilibrium, seemingly sub-optimal messages are disclosed by the sender, and the sender's equilibrium strategy favors longer messages over shorter ones, although the length of evidence has no intrinsic value in our model. Lastly, we show that increasing the expected length of evidence always decreases the sender's equilibrium value, while the effect of increasing the informativeness of signals depends on the sender's evidence endowment.

\appendix

\section{Proofs}

\subsection{Preliminaries}

We will repeatedly use Zorn's lemma in the proofs of our results. The statement of the lemma is as below. Recall that given a partially ordered set \((X,\leq)\), a \emph{chain} is a totally ordered subset of \(X\), an \emph{upper bound} of a subset \(A\subset X\) is an element \(x\in X\) such that \(x\geq y\) for all \(y\in A\), and a \emph{maximal element} of \(X\) is an element \(x\in X\) such that \(y\geq x\Rightarrow y=x\).

\begin{lmm}[Zorn's lemma] \label{lmm.a1}
	Let \(X\) be a partially ordered set. If every chain in \(X\) has an upper bound, then \(X\) has a maximal element.
\end{lmm}

By (A4$^\prime$) and Zorn's lemma, for all nonempty sets \(A\subset E\), the set of minimal elements of \(A\), i.e., \(\min(A,\precsim)=\{e\in A:(\forall e'\in A)(e'\precsim e\Rightarrow e'\sim e)\}\), is nonempty.

\subsection{Proof of Theorem \ref{thm1}} \label{sec.a1}

We divide the proof into a sequence of lemmas. Lemma \ref{lmm.a2} shows that in any equilibrium, the sender's value function satisfies conditions 1 through 3 of Theorem \ref{thm1}. Lemma \ref{lmm.a3} shows that a function \(V:E\to\mathbb{R}\) satisfying conditions 1 through 3 of Theorem \ref{thm1} is the sender's value function in an equilibrium. Lemma \ref{lmm.a4} shows that the value function is unique.

Throughout the proof, we use the following notation. Let \(\sigma:E\to\Delta(E)\) be a strategy of the sender. We denote by \(S^\sigma\) the distribution of messages induced by the sender's strategy \(\sigma\). That is, \(S^\sigma(m)=\sum_{e\in E}\sigma(m|e)F(e)\) for all \(m\in E\). For nonempty \(A\subset E\) and \(m\in E\), we similarly define the probability that message \(m\) is disclosed conditional on the sender's evidence endowment being in \(A\) as
\begin{equation*}
	S^\sigma(m|A)=\frac{1}{F(A)}\sum_{e\in A}\sigma(m|e)F(e).
\end{equation*}

\begin{lmm}[Necessity] \label{lmm.a2}
	Let \((\sigma^\star,\mathbf{a}^\star,\mu)\) be an equilibrium of the disclosure game, and \(V\) the sender's value function. Then \(V\) satisfies conditions 1 through 3 of Theorem \ref{thm1}.
\end{lmm}

\begin{proof}
	Condition 1 is by definition. If \(e_1\precsim e_2\), \(LC(e_1)\subset LC(e_2)\). Therefore, \(V(e_1)=\max_{m\in LC(e_1)}\mathbf{a}^\star(m)\leq\max_{m\in LC(e_2)}\mathbf{a}^\star(m)=V(e_2)\). That is, \(V\) is non-decreasing.

	For condition 2, let \(M_x=\{e:\xi(e)\geq V(e)=x\}\), which is the set of on-path messages at which the receiver takes action \(x\). That is, \(\phi(\mu(m))=x\) for all \(m\in M_x\). Since all sender types whose equilibrium payoff is \(x\) disclose messages in \(M_x\),
	\begin{equation*}
		\nu(V^{-1}(x))=\sum_{m\in M_x} S^{\sigma^\star}\big(m|V^{-1}(x)\big)\mu(m).
	\end{equation*}
	By in-betweenness (and continuity of \(\phi\), if \(M_x\) is countably infinite), we conclude that \(\xi(V^{-1}(x))=x\).

	For condition 3, let let \(M_L = \bigcup_{e\in L} supp(\sigma^\star(\cdot|e))\) be the set of on-path messages disclosed by sender types in \(L\). Fix any message \(m\in M_L\). Since it is disclosed only by sender types in \(V^{-1}(x)\),
	\begin{equation*}
		\mu(m) = \sum_{e\in V^{-1}(x)} \frac{\sigma^\star(m|e)F(e)}{S^{\sigma^\star}(m)}\nu(e).
	\end{equation*}
	Notice that for all \(e\in V^{-1}(x)\setminus L\), if \(\sigma^\star(m|e)>0\), disclosing \(e\) truthfully is not optimal for sender type \(e\). That is, \(\xi(e)=\phi(\nu(e))<x\). Therefore, by in-betweenness,
	\begin{equation*}
		\phi\left(\sum_{e\in L}\frac{\sigma^\star(m|e)F(e)}{S^{\sigma^\star}(m|L)F(L)}\nu(e)\right)\geq x.
	\end{equation*}
	Equality occurs in the above equation if and only if \(\sigma^\star(m|e)=0\) for all \(e\notin L\). Notice that
	\begin{equation*}
		\nu(L) = \sum_{e\in L} \frac{F(e)}{F(L)}\nu(e) = \sum_{m\in M_L} S^{\sigma^\star}(m|L)\left(\sum_{e\in L}\frac{\sigma^\star(m|e)F(e)}{S^{\sigma^\star}(m|L)F(L)}\nu(e)\right).
	\end{equation*}
	By in-betweenness, \(\xi(L)\geq x\). Equality occurs if \(\sigma^\star(m|e)=0\) for all \(m\in L\) and \(e\notin L\).
\end{proof}

\begin{lmm}[Sufficiency] \label{lmm.a3}
	Let \(V:E\to\mathbb{R}\) be a function satisfying conditions 1 through 3 of Theorem \ref{thm1}. Then there exists an equilibrium in which the sender's value function is \(V\).
\end{lmm}

\begin{proof}
	The proof is by construction. First, observe that in any equilibrium \((\sigma^\star,\mathbf{a}^\star,\mu)\), fixing any \(e\in E\), one of the following two cases happens:
	\begin{enumerate}
		\item \(e\) is an off-path message, \(\mathbf{a}^\star(e)=\xi(e)<V(e)\), and \(\mu(e)=\nu(e)\);
		\item The sender discloses truthfully if her evidence endowment is \(e\), and \(\mathbf{a}^\star(e)=V(e)\leq\xi(e)\).
	\end{enumerate}
	Therefore, any candidate equilibrium \((\sigma,\mathbf{a},\mu)\) in which the sender's value function is \(V\) satisfies the following conditions:
	\begin{enumerate}[label=(\alph*)]
		\item The set of on-path messages is \(O:=\{e:V(e)\leq\xi(e)\}\);
		\item For all \(e\in E\), \(\mathbf{a}(e)=\min\{V(e),\xi(e)\}\);
		\item For all \(m\in E\setminus O\), \(\mu(m)=\nu(m)\).
	\end{enumerate}
	The above conditions pin down the set of on-path messages, the receiver's strategy, and the receiver's off-path beliefs. The remainder of the proof constructs a strategy of the sender, \(\sigma=\{\sigma(\cdot|e)\}_{e\in E}\subset\Delta(E)\), and the receiver's beliefs at on-path messages, \(\mu|_O=\{\mu(m)\}_{m\in O}\subset\Delta(\Omega)\), such that:
	\begin{enumerate}[label=(\alph*),resume]
		\item For all \(e\in E\), \(supp(\sigma(\cdot|e))\subset LC(e)\cap V^{-1}(V(e))\cap O\);
		\item For all \(e\in O\), \(\sigma(e|e)=1\);
		\item For all \(n\) and all \(m\in O\),
		\begin{equation*}
			\mu_n(m) = \frac{\sum_{e\in E}\sigma(m|e)F_n(e)\pi_n}{\sum_{e\in E}\sigma(m|e)F(e)};
		\end{equation*}
		\item For all \(m\in O\), \(\phi(\mu(m))=V(m)\).
	\end{enumerate}
	Feasibility and sender optimality are implied by (d), truth-leaning by (e), Bayesian consistency by (f), and receiver optimality by (g). Hence, given a solution to (d) through (g), the resulted tuple \((\sigma,\mathbf{a},\mu)\) is an equilibrium in which the sender's value function is \(V\). The proof is completed.

	We now construct the sender's strategy and the receiver's beliefs at on-path messages. Let \(L\) be a lower contour set in \(E\). We say a collection of distributions \(\{\sigma(\cdot|e)\}_{e\in L}\) on \(E\) and a collection of distributions \(\{\mu(m)\}_{m\in L\cap O}\) on \(\Omega\) constitute a \emph{candidate pair} on \(L\) if they satisfy (d) through (f) on \(L\), and for all \(m\in O\cap L\), \(\phi(\mu(m))\geq V(m)\).

	If \(L=E\), we show that (g) is also satisfied. Suppose that, on the contrary, there exists \(m^\star\in O\) such that \(\phi(\mu(m^\star))>V(m^\star)\). Let \(A=V^{-1}(V(m^\star))\subset E\). Notice that
	\begin{equation*}
		\nu(A)=\sum_{m\in O\cap A} S^\sigma(m|A)\mu(m),
	\end{equation*}
	\(\phi(\mu(m))\geq V(m^\star)\) for all \(m\in O\cap A\), and \(\phi(\mu(m^\star))>V(m^\star)\). By in-betweenness, \(\phi(\nu(A))=\xi(A)>V(m^\star)\), which is a contradiction to condition 2 of Theorem \ref{thm1}. Hence, if \(L=E\), the candidate pair is a solution to (d) through (g).

	If \(L\subsetneq E\), let \(e^\star\) be a minimal element in \((E\setminus L,\precsim)\). We show below that we can extend the candidate pair to \(\tilde{L}=L\cup[e^\star]\), where \([e^\star]=\{e:e\sim e^\star\}\) is the equivalence class of \(e^\star\), in the following sense. We are to define distributions \(\sigma(\cdot|e)\) on \(E\) for all \(e\in[e^\star]\) and distributions \(\tilde{\mu}(m)\) on \(\Omega\) for all \(m\in O\cap\tilde{L}\) such that \(\{\sigma(\cdot|e)\}_{e\in\tilde{L}}\) and \(\{\tilde{\mu}(m)\}_{m\in\tilde{L}\cap O}\) constitute a candidate pair on \(\tilde{L}\). Notice that \(\tilde{\mu}\) and \(\mu\) need not be the same on \(L\cap O\).

	Let \(Q=\tilde{L}\cap V^{-1}(V(e^\star))\cap O\). First, let \(\sigma(m|e)=0\) for all \(e\in[e^\star]\) and \(m\notin Q\), and let \(\sigma(e|e)=1\) and \(\sigma(m|e)=0\) for all \(e\in[e^\star]\cap Q\) and \(m\neq e\). Thereby, (d) and (e) are satisfied on \(\tilde{L}\). Let \(\mu(m)=\nu(m)\) for all \(m\in[e^\star]\cap Q\). If \([e^\star]\subset Q\), i.e., if all of \([e^\star]\) are on-path, \(\{\sigma(\cdot|e)\}_{e\in\tilde{L}}\) is already well-defined. If \([e^\star]\not\subset Q\), let \(m_1,m_2,\dots\) be the elements of \(Q\), and we define \(\sigma(m_i|e)=z_i\) for all \(e\in[e^\star]\setminus Q\), where \(z_i\) is defined recursively as follows:
	\begin{equation*}
		z_i = \min\left\{\frac{\beta_i}{1-\beta_i}\frac{\sum_{e\in L\cup Q}\sigma(m_i|e)F(e)}{F([e^\star]\setminus Q)},1-\sum_{k=1}^{i-1}z_k\right\},
	\end{equation*}
	with the initial condition
	\begin{equation*}
		z_1 = \min\left\{\frac{\beta_1}{1-\beta_1}\frac{\sum_{e\in L\cup Q}\sigma(m_1|e)F(e)}{F([e^\star]\setminus Q)},1\right\},
	\end{equation*}
	where
	\begin{equation*}
		\beta_i = \max\left\{\beta:\phi\big(\beta\nu([e^\star]\setminus Q)+(1-\beta)\mu(m_i)\big)\geq V(e^\star)\right\}.
	\end{equation*}
	By in-betweenness, \(\beta_i\geq 0\), and by continuity of \(\phi\), \(\phi\big(\beta_i\nu([e^\star]\setminus Q)+(1-\beta_i)\mu(m_i)\big)=V(e^\star)\) for all \(i\). Intuitively, \(\beta_i\) measures the distance between \(\mu(m_i)\) and the ``worst'' belief at which choosing action \(V(e^\star)\) is receiver optimal. If \(\beta_i>0\), then sender types in \([e^\star]\setminus Q\) can disclose message \(m_i\) with positive probability. We are left to show that \(\sigma(\cdot|e)\) is well-defined for \(e\in[e^\star]\setminus Q\) through this way. That is, we need to show that \(\sum_i z_i=1\). Suppose not, i.e., \(\sum_i z_i<1\). Let \(A=V^{-1}(V(e^\star))\). Then
	\begin{align*}
		\nu(\tilde{L}\cap A) &= \frac{F([e^\star]\setminus Q)}{F(\tilde{L})} \nu([e^\star]\setminus Q)+\sum_i\lambda_i\mu(m_i) \\
		&= \sum_i \frac{\lambda_i}{1-\beta_i}\big(\beta_i\nu([e^\star]\setminus Q)+(1-\beta_i)\mu(m_i)\big) + \left(1-\sum_i z_i\right)\frac{F([e^\star]\setminus Q)}{F(\tilde{L})}\nu([e^\star]\setminus Q),
	\end{align*}
	where
	\begin{equation*}
		\lambda_i = \frac{\sum_{e\in L\cup Q}\sigma(m_i|e)F(e)}{F(\tilde{L})}.
	\end{equation*}
	By in-betweenness, \(\phi(\nu(\tilde{L}\cap A))=\xi(\tilde{L}\cap A)<V(e^\star)\). But \(\tilde{L}\cap A\) is a lower contour set in \(A\), so this contradicts condition 3 of Theorem \ref{thm1}. Hence, \(\{\sigma(\cdot|e)\}_{e\in\tilde{L}}\) is well-defined.

	We now define \(\tilde{\mu}:\tilde{L}\cap O\to\Delta{\Omega}\) such that
	\begin{equation}\label{eq.x1}
		\tilde{\mu}_n(m) = \frac{\sum_{e\in\tilde{L}}\sigma(m|e)F_n(e)\pi_n}{\sum_{e\in\tilde{L}}\sigma(m|e)F(e)}
	\end{equation}
	for all \(n\) and all \(m\in\tilde{L}\cap O\). Then (f) is satisfied on \(\tilde{L}\). We are left to check that \(\phi(\tilde{\mu}(m))\geq V(m)\) for all \(m\in O\cap\tilde{L}\). Notice that \(\tilde{\mu}(m)=\mu(m)\) for all \(m\notin Q\), so we only need to check this for each \(m_i\in Q\). Notice that \(\tilde{\mu}(m_i)=\alpha_i\nu([e^\star]\setminus Q)+(1-\alpha_i)\mu(m_i)\), where
	\begin{equation*}
		\alpha_i = \frac{z_i F([e^\star]\setminus Q)}{z_i F([e^\star]\setminus Q)+\sum_{e\in L\cup Q}\sigma(m_i|e)F(e)}.
	\end{equation*}
	By construction, \(\alpha_i\leq\beta_i\), so \(\phi(\tilde{\mu}(m_i))\geq V(e^\star)=V(m_i)\). Hence, \(\{\sigma(\cdot|e)\}_{e\in\tilde{L}}\) and \(\{\tilde{\mu}(m)\}_{m\in\tilde{L}\cap O}\) constructed above constitute a candidate pair on \(\tilde{L}\).

	We can now proceed with transfinite induction. Let \(\mathcal{X}\) be the space of all candidate pairs on some lower contour set in \(E\). Let \(x_1=(\{\sigma_1(\cdot|e)\}_{e\in L_1},\{\mu_1(m)\}_{m\in L_1\cap O})\) and \(x_2=(\{\sigma_2(\cdot|e)\}_{e\in L_2},\{\mu_2(m)\}_{m\in L_2\cap O})\) be two typical elements in \(\mathcal{X}\). We can define a partial order, denoted \(\sqsubset\), on \(\mathcal{X}\), such that \(x_1\sqsubset x_2\) if and only if \(L_1\subset L_2\) and \(\sigma_1(\cdot|e)=\sigma_2(\cdot|e)\) for all \(e\in L_1\). For every chain \(\mathscr{C}=\{x_\alpha\}_{\alpha\in\mathscr{I}}\) in \(\mathcal{X}\), it has an upper bound \(\tilde{x}=(\{\tilde{\sigma}(\cdot|e)\}_{e\in\tilde{L}},\{\tilde{\mu}(m)\}_{m\in\tilde{L}\cap O})\), where \(\tilde{L}=\cup_{\alpha\in\mathscr{I}}L_\alpha\), \(\tilde{\sigma}(\cdot|e)=\sigma_\alpha(\cdot|e)\) for some \(\alpha\in\mathscr{I}\) such that \(L_\alpha\ni e\)\footnote{Since \(\mathcal{X}\) is totally ordered, for all \(\alpha,\alpha'\in\mathscr{I}\), \(x_\alpha\sqsubset x_{\alpha'}\) or \(x_{\alpha'}\sqsubset x_\alpha\). Therefore, if \(e\) is in both \(L_\alpha\) and \(L_{\alpha'}\), \(\sigma_\alpha(\cdot|e)=\sigma_{\alpha'}(\cdot|e)\). We can thus define \(\tilde{\sigma}(\cdot|e)\) using any \(\alpha\in\mathscr{I}\) such that \(L_\alpha\ni e\).}, and \(\tilde{\mu}(m)\) is defined by (\ref{eq.x1}). By construction, \(\tilde{x}\) satisfies (d) through (f). By continuity of \(\phi\), \(\phi(\tilde{\mu}(m))\geq V(m)\) for all \(m\in O\cap\tilde{L}\). Hence, \(\tilde{x}\in\mathcal{X}\) is a candidate pair on \(\tilde{L}\). By Zorn's lemma, \(\mathcal{X}\) has a maximal element. Let \(\hat{x}\) be one such maximal element, and suppose that it is a candidate pair on some lower contour set \(\hat{L}\) in \(E\). If \(\hat{L}\neq E\), we have shown that we can find extend \(\hat{x}\) to a superset of \(\hat{L}\), which violates maximality. Hence, \(\hat{L}=E\), and by the earlier argument, \(\hat{x}\) is a solution to (d) through (g).

	It remains to show that \(\mathcal{X}\) is nonempty, i.e., there exists a candidate pair. Notice that the above extension process applies when \(L=\emptyset\). When \(L=\emptyset\), let \(e^\star\) be a minimal element of \(E\), and \(Q=[e^\star]\cap O\) is nonempty by condition 2 of Theorem \ref{thm1}. Following the extension process, we can construct a candidate pair on \([e^\star]\). Hence, \(\mathcal{X}\) is nonempty.
\end{proof}

\begin{lmm}[Uniqueness] \label{lmm.a4}
	The function that satisfies conditions 1 through 3 of Theorem \ref{thm1} is unique.
\end{lmm}

\begin{proof}
	Suppose, contrary to our claim, that \(V\neq W\) are two functions that satisfy conditions 1 through 3 of Theorem \ref{thm1}. Without loss of generality, there exist \(e_1,e_2\in E\) such that \(V(e_1)=V(e_2)\) and \(W(e_1)<W(e_2)\).\footnote{Otherwise, either for all \(e_1\neq e_2\), \(V(e_1)\neq V(e_2)\) and \(W(e_1)\neq W(e_2)\), or \(V(e_1)=V(e_2)\Leftrightarrow W(e_1)=W(e_2)\). In the former case, \(V(e)=W(e)=\xi(e)\) for all \(e\in E\), which is a contradiction. In the latter case, for all \(x\in Range(V)\), \(\exists y\in Range(W)\) such that \(V^{-1}(x)=W^{-1}(y)\) (and vice versa). But by condition 2 of Theorem \ref{thm1}, \(x=\xi(V^{-1}(x))=\xi(V^{-1}(y))=y\). Hence, \(V=W\), which is a contradiction.} Let
	\begin{equation*}
		L_w = \bigcup_{x\leq W(e_1)}W^{-1}(x),~U_w = \bigcup_{x\geq W(e_2)}W^{-1}(x),
	\end{equation*}
	\begin{equation*}
		L_v = \bigcup_{y\leq V(e_1)}V^{-1}(y),~U_v = \bigcup_{y\geq V(e_2)}V^{-1}(y).
	\end{equation*}
	By condition 1 of Theorem \ref{thm1}, \(L_w\) is a lower contour set in \(E\). Therefore, for all \(y\geq V(e_2)\) and \(y\in Range(V)\), \(V^{-1}(y)\cap L_w\) is a lower contour set in \(V^{-1}(y)\) (or empty). By condition 3 of Theorem \ref{thm1}, \(\xi(V^{-1}(y)\cap L_w)\geq y\geq V(e_2)\) for all \(y\) such that \(V^{-1}(y)\cap L_w\neq\emptyset\). Hence, \(\xi(U_v\cap L_w)\geq V(e_2)\). But \(U_v\) is an upper contour set in \(E\). Therefore, for all \(x\leq W(e_1)\) and \(x\in Range(W)\), \(U_v\cap W^{-1}(x)\) is an upper contour set in \(W^{-1}(x)\) (or empty), so \(\xi(U_v\cap W^{-1}(x))\leq x\leq W(e_1)\) for all \(x\) such that \(U_v\cap W^{-1}(x)\neq\emptyset\). Hence, \(\xi(U_v\cap L_w)\leq W(e_1)\). That is,
	\begin{equation} \label{eq.a8}
		V(e_2)\leq\xi(U_v\cap L_w)\leq W(e_1).
	\end{equation}
	Similarly, we can show that
	\begin{equation} \label{eq.a9}
		W(e_2)\leq\xi(U_w\cap L_v)\leq V(e_1).
	\end{equation}
	The inequalities in (\ref{eq.a8}) and (\ref{eq.a9}) are in contradiction with the assumptions that \(W(e_1)<W(e_2)\) and \(V(e_1)=V(e_2)\). Hence, the function that satisfies conditions 1 through 3 of Theorem \ref{thm1} is unique.
\end{proof}

\subsection{Proof of Corollary \ref{crl2}}

\begin{proof}
	Given an equilibrium and the sender's value function \(V\), we can define an ordered partition \((\mathcal{A},\leq_P)\) of \(E\), where \(\mathcal{A}=\{V^{-1}(x):x\in Range(V)\}\), and \(A_1\leq_P A_2\Leftrightarrow\xi(A_1)\leq\xi(A_2)\). By construction, the ordered partition \((\mathcal{A},\leq_P)\) satisfies conditions 1 and 3 of Corollary \ref{crl2}. By Theorem \ref{thm1}, for all \(A_1,A_2\in\mathcal{A}\) and \(e_1\in A_1,e_2\in A_2\) such that \(e_1\precsim e_2\), \(\xi(A_1)=V(e_1)\leq V(e_2)=\xi(A_2)\). That is, \(A_1\leq_P A_2\), and condition 2 of Corollary \ref{crl2} is satisfied.

	Given any ordered partition \((\mathcal{A},\leq_P)\) of \(E\) that satisfies conditions 1 through 3 of Corollary \ref{crl2}, we can define \(V:E\to[0,1]\) such that \(V(e)=\xi(A)\) for all \(A\in\mathcal{A}\) and \(e\in A\). By construction, \(V\) satisfies conditions 1 through 3 of Theorem \ref{thm1}, so it is the sender's value function in an equilibrium.

	Let \((\mathcal{A},\leq_P)\) and \((\mathcal{A}^\star,\leq_P^\star)\) both satisfy conditions 1 through 3 of Corollary \ref{crl2}. Suppose that \(\mathcal{A}\neq\mathcal{A}^\star\). Then, without loss of generality, there exist \(e,e'\in E\), \(A,A'\in\mathcal{A}\), \(A\neq A'\), and \(A^\star\in\mathcal{A}^\star\) such that \(e\in A\), \(e'\in A'\), \(e,e'\in A^\star\). Hence, \(V(e)=V(e')=\xi(A^\star)\). But \(V(e)=\xi(A)\neq \xi(A')=V(e')\). Contradiction reached. Hence, \(\mathcal{A}=\mathcal{A}^\star\). By condition 1 of Corollary \ref{crl2}, \(A_1\leq_P A_2\Leftrightarrow \xi(A_1)\leq\xi(A_2)\Leftrightarrow A_1\leq_P^\star A_2\) for all \(A_1,A_2\in\mathcal{A}=\mathcal{A}^\star\). Hence, \(\leq_P\) is also identical to \(\leq_P^\star\). That is, \((\mathcal{A},\leq_P)=(\mathcal{A}^\star,\leq_P^\star)\). Uniqueness is shown.
\end{proof}

\subsection{Proof of Theorem \ref{thm3}} \label{sec.a4}
\begin{proof}
	For the ``if'' part, let \((\mathcal{A},\leq_P)\) be an ordered partition of \(E\) such that every \(A\in\mathcal{A}\) is the largest minimizer of (\ref{eq2.7}). We show that \((\mathcal{A},\leq_P)\) satisfies conditions 1 through 3 of Corollary \ref{crl2} and therefore is the equilibrium partition. Suppose that condition 1 is violated. Then there exist \(A_1 <_P A_2\) such that \(\xi(A_1)\geq\xi([A_1,A_2])\), where \([A_1,A_2]=\bigcup\{A':A_1 \leq_P A' \leq_P A_2\}\). This contradicts the assumption that \(A_1\) is the largest minimizer of (\ref{eq2.7}). Hence, condition 1 is satisfied. Let \(A_1,A_2\in\mathcal{A}\), \(e_1\in A_1\), \(e_2\in A_2\) such that \(e_1\precsim e_2\). Since \(A_2\) is a lower contour set in \(E\setminus\cup_{A'<_P A_2}A'\), \(e_1\in A_2\) or \(e_1\in\cup_{A'<_P A_2}A'\). Hence, \(A_1\leq_P A_2\); condition 2 is satisfied. Let \(A\in\mathcal{A}\), and let \(L\) be a lower contour set in \(A\). Since \(L\) is a lower contour set in \(E\setminus\cup_{A'<_P A}A'\), and \(A\) solves (\ref{eq2.7}), \(\xi(L)\geq\xi(A)\); condition 3 is satisfied.

	For the ``only if'' part, suppose that \((\mathcal{A},\leq_P)\) is the equilibrium partition, we are to show that every \(A\in\mathcal{A}\) is the largest minimizer of (\ref{eq2.7}). To obtain a contradiction, consider two separate cases. First, suppose that for some \(A_0\in\mathcal{A}\), \(A_0\) does not solve (\ref{eq2.7}). Then there exists a lower contour set \(L\) in \(E\setminus\bigcup_{A'<_P A_0}A'\) such that \(\xi(L)<\xi(A_0)\). Notice that \(\{A\cap L\}_{A\geq_P A_0}\) is a partition of \(L\), and by assumption, \(\xi(A\cap L)\geq\xi(A)\geq\xi(A_0)\) for all \(A\geq_P A_0\) such that \(A\cap L\neq\emptyset\). Hence, \(\xi(L)\geq\xi(A_0)\), a contradiction. Second, suppose that there exists \(A_0\in\mathcal{A}\) that is not the largest minimizer of (\ref{eq2.7}). Then there exists \(A_0'\not\subset A_0\) that solves (\ref{eq2.7}). Since \(A_0'\setminus A_0\) is an upper contour set in \(A_0'\), \(\xi(A_0'\setminus A_0)\leq\xi(A_0')=\xi(A_0)\). But \(\{A\cap A_0'\}_{A>_P A_0}\) is a partition of \(A_0'\setminus A_0\), and by assumption, \(\xi(A\cap A_0')\geq\xi(A)>\xi(A_0)\) for all \(A>_P A_0\) such that \(A\cap A_0'\neq\emptyset\). Hence, \(\xi(A_0'\setminus A_0)>\xi(A_0)\), a contradiction.
\end{proof}

\subsection{Results Relating to Equilibrium Existence} \label{sec.a5}
\begin{lmm} \label{lmm.a5}
	Let \(A\subset E\), and suppose that
	\begin{equation} \label{eq.a10}
		\min_{L} \xi(L) \quad s.t. \quad \text{\(L\) is a lower contour set in \(A\)}
	\end{equation}
	admits a solution. Then (\ref{eq.a10}) has a largest minimizer.
\end{lmm}

\begin{proof}
	Let \(\mathcal{L}\) be the set of minimizers of (\ref{eq.a10}), and \(\underline{v}\) the minimum of (\ref{eq.a10}). For every chain \(\{L_\alpha\}_{\alpha\in\mathscr{I}}\) in \((\mathcal{L},\subset)\), let \(\bar{L}=\cup_{\alpha\in\mathscr{I}}L_\alpha\). \(\bar{L}\) is an upper bound of \(\{L_\alpha\}_{\alpha\in\mathscr{I}}\). Moreover, \(\bar{L}\) is a lower contour set in \(A\), and \(\xi(\bar{L})=\underline{v}\). Hence, \(\bar{L}\in\mathcal{L}\). By Zorn's lemma, there exists a maximal element of \((\mathcal{L},\subset)\). Let \(L^\star\) be one maximal element. Suppose that \(L^\star\) is not the greatest element in \((\mathcal{L},\subset)\), then there exists \(L'\in\mathcal{L}\) such that \(L'\not\subset L^\star\). \(\xi(L'\cup L^\star)=\underline{v}\), so \(L'\cup L^\star\in\mathcal{L}\). This is a contradiction to maximality of \(L^\star\). Hence, \(L^\star\) is the largest minimizer of (\ref{eq.a10}).
\end{proof}

\begin{lmm} \label{lmm.a6}
	Let \(A\subset E\) have finitely many minimal elements. Then (\ref{eq.a10}) has a largest minimizer.
\end{lmm}

\begin{proof}
	Let \(e\) be a minimal element in \(A\), and consider the problem
	\begin{equation} \label{eq.a11}
		\inf_{L}\xi(L) \quad s.t. \quad \text{\(L\ni e\) and \(L\) is a lower contour set in \(A\)}.
	\end{equation}
	Let \(\underline{v}\) be the infimum of (\ref{eq.a11}), and \(\{L_n\}_{n=1}^\infty\) a sequence of lower contour sets in \(A\) containing \(e\) such that \(\xi(L_n)\downarrow\underline{v}\). Equivalently, \(\nu(L_n)\downarrow\phi^{-1}(\underline{v})\). For all \(n\), \(\nu(\bigcup_{m\geq n}L_m)\leq\nu(L_n)\), so \(\nu(\bigcup_{m\geq n}L_m)\to\phi^{-1}(\underline{v})\) as \(n\to\infty\). Let \(\bar{L}=\bigcap_{n=1}^\infty\bigcup_{m\geq n}L_m\). It is a lower contour set in \(A\). Moreover, \(\nu(\bigcup_{m\geq n}L_m)\to\nu(\bar{L})\), so \(\nu(\bar{L})=\phi^{-1}(\underline{v})\), i.e., \(\xi(\bar{L})=\underline{v}\). Hence, \(\bar{L}\) is a minimizer of (\ref{eq.a11}). Since there are only finitely many minimal elements in \(A\), and any lower contour set in \(A\) contains at least one minimal element of \(A\), (\ref{eq.a10}) admits a solution. By Lemma \ref{lmm.a5}, (\ref{eq.a10}) has a largest minimizer.
\end{proof}

\subsection{Proof of Theorem \ref{thm4}}
\begin{proof}
	Without loss of generality, let \(E\) be connected. Let \(\mathscr{A}\) be the collection of all \(\mathcal{A}\) such that \(\mathcal{A}\) is a partition of some lower contour set in \(E\) consisting of connected sets and satisfying the following:
	\begin{enumerate} [label=(\alph*)]
		\item Every \(A\in\mathcal{A}\) is a lower contour set in \(E\setminus\bigcup\{A'\in\mathcal{A}:\xi(A')<\xi(A)\}\);
		\item For all \(A\in\mathcal{A}\) and all connected lower contour sets \(K\) in \(E\setminus\bigcup\{A'\in\mathcal{A}:\xi(A')<\xi(A)\}\) such that \(A\subset K\), \(\xi(K)>\xi(A)\);
		\item For all \(A\in\mathcal{A}\) and all lower contour sets \(L\) in \(A\), \(\xi(L)\geq\xi(A)\).
	\end{enumerate}
	By Lemma \ref{lmm.a6}, \(\mathscr{A}\) is nonempty, since \(\{A_0\}\in\mathscr{A}\), where \(A_0\) is the largest lower contour set in \(E\) that minimizes \(\xi\). Moreover, \(\mathscr{A}\) is partially ordered by inclusion \(\subset\), and every chain \(\mathscr{C}=\{\mathcal{A}_\alpha\}_{\alpha\in\mathscr{I}}\) in \((\mathscr{A},\subset)\) has an upper bound \(\cup_{\alpha\in\mathscr{I}}\mathcal{A}_\alpha\). By Zorn's lemma, \((\mathscr{A},\subset)\) has a maximal element. Let \(\mathcal{A}^\star\) be one maximal element of \((\mathscr{A},\subset)\).

	We show that \(\bigcup\mathcal{A}^\star = E\). Suppose not, then there exists a partition \(\mathcal{P}\) of \(E\setminus\bigcup\mathcal{A}^\star\) such that every \(P\in\mathcal{P}\) is connected, and for all distinct \(P_1,P_2\in\mathcal{P}\), \(P_1\cup P_2\) is not connected.\footnote{In other words, \(\mathcal{P}\) is the coarsest partition of \(E\setminus\bigcup\mathcal{A}^\star\) consisting of connected sets.} By assumption, every \(P\in\mathcal{P}\) has finitely many minimal elements. By Lemma \ref{lmm.a6}, we can find \(A(P)\) for every \(P\in\mathcal{P}\) such that \(A(P)\) is the largest lower contour set in \(P\) that minimizes \(\xi\). Let \(\mathcal{A}'=\mathcal{A}^\star\cup\{A(P):P\in\mathcal{P}\}\). We verify that \(\mathcal{A}'\in\mathscr{A}\). Since \(E\setminus\bigcup\{A'\in\mathcal{A}':\xi(A')<\xi(A)\}\subset E\setminus\bigcup\{A'\in\mathcal{A}^\star:\xi(A')<\xi(A)\}\), and \(\mathcal{A}^\star\in\mathscr{A}\), conditions (a) through (c) are satisfied if \(A\in\mathcal{A}^\star\). If \(A=A(P)\) for some \(P\in\mathcal{P}\), then \(P\) is a maximal connected subset of \(E\setminus\bigcup\{A'\in\mathcal{A}':\xi(A')<\xi(A)\}\).\footnote{That is, \(P\subset E\setminus\bigcup\{A'\in\mathcal{A}':\xi(A')<\xi(A)\}\), and \(\nexists Q\subset E\setminus\bigcup\{A'\in\mathcal{A}':\xi(A')<\xi(A)\}\) such that \(Q\) is connected and \(P\subsetneq Q\). Suppose that this is not true, then either (1) there exists \(\hat{A}\in\mathcal{A}'\) such that \(\xi(\hat{A})<\xi(A)\) and \(P\cap\hat{A}\neq\emptyset\), or (2) there exists \(\hat{A}\in\mathcal{A}'\) such that \(\xi(\hat{A})\geq\xi(A)\) and \(P\cup\hat{A}\) is connected. (1) is impossible, since \(P\cap(\bigcup\mathcal{A}^\star)=\emptyset\), and for all \(P'\neq P\) in \(\mathcal{P}\), \(P\cap A(P')\subset P\cap P'=\emptyset\). (2) is also impossible. If \(\hat{A}\in\mathcal{A}^\star\), then \(A\cup\hat{A}\) is a connected lower contour set in \(E\setminus\{A'\in\mathcal{A}^\star:\xi(A')<\xi(\hat{A})\}\), and \(\xi(A\cup\hat{A})\leq\xi(\hat{A})\), which contradicts \(\hat{A}\in\mathcal{A}^\star\in\mathscr{A}\). If \(\hat{A}=A(P')\) for some \(P'\in\mathcal{P}\), then \(P\cup P'\) is connected, which is a contradiction to the definition of \(\mathcal{P}\).} Moreover, \(P\) is a lower contour set in \(E\setminus\bigcup\{A'\in\mathcal{A}':\xi(A')<\xi(A)\}\).\footnote{Otherwise, there exists \(P'\neq P\) in \(\mathcal{P}\), \(e\in P\), and \(e'\in P'\), such that \(e'\precsim e\). Then \(P\cup P'\) is connected, which is a contradiction.} Thus, by construction, conditions (a) through (c) are satisfied. Hence, \(\mathcal{A}'\in\mathscr{A}\), and \(\mathcal{A}'\supsetneq\mathcal{A}^\star\). This is a contradiction to maximality of \(\mathcal{A}^\star\). Therefore, \(\bigcup\mathcal{A}^\star = E\).

	Finally, we construct the equilibrium partition from \(\mathcal{A}^\star\). Let \(\mathcal{A}^{\star\star}=\big\{\bigcup\{A'\in\mathcal{A}^\star:\xi(A')=\xi(A)\}:A\in\mathcal{A}^\star\big\}\). That is, \(\mathcal{A}^{\star\star}\) is a coarser partition of \(E\) than \(\mathcal{A}^\star\), in which we take unions over sets in \(\mathcal{A}^\star\) that have a same value of \(\xi\). On \(\mathcal{A}^{\star\star}\), \(\xi\) takes distinct values, which is not necessarily the case on \(\mathcal{A}^\star\). Hence, \(A_1\leq_P A_2\Leftrightarrow \xi(A_1)\leq \xi(A_2)\) defines a total order \(\leq_P\) on \(\mathcal{A}^{\star\star}\). Moreover, by construction, every \(A\in\mathcal{A}^{\star\star}\) is the largest lower contour set in \(E\setminus\bigcup_{A'<_P A}A'\) that minimizes \(\xi\). Therefore, by Theorem \ref{thm3}, an equilibrium exists, and \((\mathcal{A}^{\star\star},\leq_P)\) is the equilibrium partition.
\end{proof}

\subsection{Proof of Proposition \ref{pro4}}
\begin{proof}
	Since \(M\) is non-decreasing and \(p>q\), \(V\) is non-decreasing. Let \(E_n=\{e\in E:M(e)=n\}\) for every \(n\in\mathbb{N}\). For all \(x\in Range(V)\), \(V^{-1}(x)=E_n\) for some \(n\in\mathbb{N}\). We are to verify conditions 2 and 3 of Theorem \ref{thm1} for every \(E_n\).

	Notice that \(E_0=\{\emptyset\}\cup\{(b,h):h\in E_0\cup E_1\}\), and \(E_n=\{(g,h):h\in E_{n-1}\}\cup\{(b,h):h\in E_{n+1}\}\) for all \(n\geq 1\). Hence, for \(j=1,2\),
	\begin{gather*}
		F_j(E_0)=(1-p-q)+p^{\mathbb{I}_{\{j=1\}}}q^{\mathbb{I}_{\{j=2\}}}(F_j(E_0)+F_j(E_1)), \\
		F_j(E_n)=p^{\mathbb{I}_{\{j=2\}}}q^{\mathbb{I}_{\{j=1\}}} F_j(E_{n-1})+p^{\mathbb{I}_{\{j=1\}}}q^{\mathbb{I}_{\{j=2\}}} F_j(E_{n+1})~\text{for all \(n\geq 1\)},
	\end{gather*}
	and
	\begin{equation*}
		\sum_{n=0}^\infty F_j(E_n)=1.
	\end{equation*}
	It follows that
	\begin{equation*}
		F_1(E_n)=\frac{p-\alpha}{p}\left(\frac{\alpha}{p}\right)^n; ~
		F_2(E_n)=\frac{q-\alpha}{q}\left(\frac{\alpha}{q}\right)^n.
	\end{equation*}
	Hence, for all \(n\in\mathbb{N}\) and \(e\in E_n\),
	\begin{equation*}
		\xi(E_n)=\nu_2(E_n)=\frac{1}{1+\frac{\pi_1}{\pi_2}\frac{p-\alpha}{q-\alpha}\left(\frac{q}{p}\right)^{n+1}}=V(e).
	\end{equation*}
	Condition 2 is satisfied.

	We now verify condition 3. Notice that for all \(n\in\mathbb{N}\) and \(e\in E_n\), \(e\) can be uniquely decomposed into two parts, as follows:
	\begin{equation*}
		e=(e_0,e_n^{\min}),
	\end{equation*}
	where \(e_0\in E_0\), and \(e_n^{\min}\in\min(E_n,\precsim)\). Therefore, every lower counter set \(L\) in \(E_0\) can be partitioned into at most countably many subsets, each having the form
	\begin{equation}\label{eq.c1}
		\{(e_0,e_n^{\min}):e_0\in L_0\},
	\end{equation}
	where \(e_n^{\min}\in\min(E_n,\precsim)\), and \(L_0\) is a lower contour set in \(E_0\). Hence, it suffices to show that, for all sets \(L'\) of the form (\ref{eq.c1}), \(\xi(L')=\nu_2(L')\geq\nu_2(E_n)=\xi(E_n)\). Notice that
	\begin{equation*}
		\nu_2(L')=\frac{1}{1+\frac{\pi_1}{\pi_2}\frac{F_1(L_0)}{F_2(L_0)}\left(\frac{q}{p}\right)^n},
	\end{equation*}
	and
	\begin{equation*}
		\nu_2(E_n)=\frac{1}{1+\frac{\pi_1}{\pi_2}\frac{p-\alpha}{q-\alpha}\left(\frac{q}{p}\right)^{n+1}}.
	\end{equation*}
	It is equivalent to show that \(\nu_2(L_0)\geq\nu_2(E_0)\) for all lower contour sets \(L_0\) in \(E_0\).

	Suppose that this is not true. Then for some small \(\varepsilon>0\), there exists a lower contour set \(L_0\) in \(E_0\) such that \(\nu_2(L_0)\leq\nu_2(E_0)-\varepsilon\). Without loss of generality, choose
	\begin{equation*}
		\varepsilon<\nu_2(E_0)-\frac{1}{1+\frac{\pi_1}{\pi_2}\frac{p-\alpha}{q-\alpha}}.
	\end{equation*}
	Denote by \(\mathcal{L}_\varepsilon\) the nonempty set of all lower contour sets \(L_0\) in \(E_0\) such that \(\nu(L_0)\leq\nu(E_0)-\varepsilon\). Notice that every chain \(\{L_\alpha\}_{\alpha\in\mathscr{I}}\) in \((\mathcal{L}_\varepsilon,\subset)\) has an upper bound \(\cup_{\alpha\in\mathscr{I}}L_\alpha\). Therefore, there exists a maximal element \(L_\varepsilon^\star\) in \(\mathcal{L}_\varepsilon\) by Zorn's lemma. To obtain a contradiction, let us discuss the following two cases.

	If there exists \(e^\star\in\min(E_0\setminus L_\varepsilon^\star,\precsim)\) such that \(D(e^\star)<0\), then \(\tilde{L}_\varepsilon := L_\varepsilon^\star \cup \{(e_0,e^\star):e_0\in E_0\}\) is a lower contour set in \(E_0\), and \(\nu_2(\tilde{L}_\varepsilon)\leq\nu_2(E_0)-\varepsilon\), since
	\begin{equation*}
		\nu_2(\{(e_0,e^\star):e_0\in E_0\}) = \frac{1}{1+\frac{\pi_1}{\pi_2}\frac{p-\alpha}{q-\alpha}\left(\frac{q}{p}\right)^{D(e^\star)+1}} \leq \frac{1}{1+\frac{\pi_1}{\pi_2}\frac{p-\alpha}{q-\alpha}} < \nu_2(E_0)-\varepsilon.
	\end{equation*}
	Hence, \(\tilde{L}_\varepsilon\in\mathcal{L}_\varepsilon\). But \(L_\varepsilon^\star\subsetneq\tilde{L}_\varepsilon\), a contradiction to maximality of \(L_\varepsilon^\star\).

	If \(D(e^\star)=0\) for all \(e^\star\in\min(E_0\setminus L_\varepsilon^\star,\precsim)\), then
	\begin{equation*}
		E_0\setminus L_\varepsilon^\star = \bigcup_{e^\star\in\min(E_0\setminus L_\varepsilon^\star,\precsim)} \{(e_0,e^\star):e_0\in E_0\}.
	\end{equation*}
	Since \(\nu_2(\{(e_0,e^\star):e_0\in E_0\})=\nu_2(E_0)\) for all \(e^\star\in\min(E_0\setminus L_\varepsilon^\star,\precsim)\), \(\nu_2(E_0\setminus L_\varepsilon^\star)=\nu_2(E_0)\). Hence, \(\nu_2(L_\varepsilon^\star)=\nu_2(E_0)\), a contradiction.

	Therefore, \(\nu_2(L_0)\geq\nu_2(E_0)\) for all lower contour sets \(L_0\) in \(E_0\), and we verify condition 3.
\end{proof}

\subsection{Proof of Proposition \ref{pro5}}
\begin{proof}
	1 is shown in the proof of Lemma \ref{lmm.a3}, and 2 is by calculation found in the main text. 3 is clearly true if for sender type \(e\), it is optimal to disclose truthfully. Hence, the proof is concluded once we show that 3 holds for evidence endowment \(\hat{e}\) such that \(\xi(\hat{e})<V(\hat{e})\). Let \(\hat{e}\in E_n\), and let \(0\leq k_1<k_2<\dots<k_l=k^\star\) be the elements in \(\argmax_kD(\hat{e}|_k)\). For each \(1\leq i\leq l\), let \(K_i:=\{(e_0^-,\hat{e}|_{k_i}):e_0^-\in E_0^-\}\), where \(E_0^-=\{\emptyset\}\cup\{e:D(e|_k)<0,~\forall 0<k\leq L(e)\}\). It is equivalent to show that \(supp(\sigma^\star(\cdot|\hat{e}))\subset K_l\), so that \(\hat{e}|_{k^\star}\precsim m\) for all \(m\in supp(\sigma^\star(\cdot|\hat{e}))\). We show this by induction.

	First, notice that for all \(e\subset K_1\), sender optimality requires that \(supp(\sigma^\star(\cdot|e))\subset K_1\). Now, assume that for some \(i<l\), we have \(supp(\sigma^\star(\cdot|e))\subset K_j\) for all \(e\in K_j\) and \(1\leq j\leq i\). Suppose that there exist \(e\in K_{i+1}\), \(j\leq i\), and \(m\in K_j\), such that \(\sigma^\star(m|e)>0\), and argue by contradiction. Notice that for all \(m\in K_j\), \(\mu_2(m)=V(\hat{e})\), so
	\begin{equation*}
		V(\hat{e})=\frac{\left(\sum_{e\in K_j}\sum_{m\in K_j}+\sum_{e\notin K_j}\sum_{m\in K_j}\right)\left(\sigma^\star(m|e)F_2(e)\pi_2\right)}{\left(\sum_{e\in K_j}\sum_{m\in K_j}+\sum_{e\notin K_j}\sum_{m\in K_j}\right)\left(\sigma^\star(m|e)F(e)\right)}.
	\end{equation*}
	By assumption,
	\begin{equation*}
		\frac{\sum_{e\in K_j}\sum_{m\in K_j}\sigma^\star(m|e)F_2(e)\pi_2}{\sum_{e\in K_j}\sum_{m\in K_j}\sigma^\star(m|e)F(e)}=\nu_2(K_j)=V(\hat{e}),
	\end{equation*}
	and
	\begin{equation*}
		\frac{\sum_{e\notin K_j}\sum_{m\in K_j}\sigma^\star(m|e)F_2(e)\pi_2}{\sum_{e\notin K_j}\sum_{m\in K_j}\sigma^\star(m|e)F(e)}<V(\hat{e}).
	\end{equation*}
	Therefore, \(V(\hat{e})<V(\hat{e})\), which is a contradiction. Hence, it must be the case that for all \(e\in K_{i+1}\), \(supp(\sigma^\star(\cdot|e))\subset K_{i+1}\). By induction, this holds for \(K_l\). Since \(\hat{e}\in K_l\), \(supp(\sigma^\star(\cdot|\hat{e}))\subset K_l\).
\end{proof}

\subsection{Proof of Proposition \ref{pro7}}
\begin{proof}
	For 1, notice that
	\begin{equation} \label{eq.c4}
		\begin{aligned}
			\frac{p-\alpha}{q-\alpha} &= \frac{2p-1+\sqrt{1-4pq}}{2q-1+\sqrt{1-4pq}} \\
			&= \frac{(\sqrt{1-4pq}+2p-1)(\sqrt{1-4pq}-2q+1)}{(1-4pq)-(2q-1)^2} \\
			&= \frac{2(p+q)-8pq+2(p-q)\sqrt{1-4pq}}{4q(1-p-q)} \\
			&= \frac{1}{2}(\gamma+1)\kappa-2\frac{\gamma}{\gamma+1}(\kappa-1)+\frac{1}{2}(\gamma-1)\sqrt{\Delta} \\
			&= \frac{(\gamma-1)^2}{2(\gamma+1)}\kappa + \frac{2\gamma}{\gamma+1} + \frac{1}{2}(\gamma-1)\sqrt{\Delta},
		\end{aligned}
	\end{equation}
	where
	\begin{equation*}
		\Delta = \kappa^2-\frac{4\gamma}{(\gamma+1)^2}(\kappa-1)^2
	\end{equation*}
	is increasing in \(\kappa\). Hence, (\ref{eq.c4}) is increasing in \(\kappa\), and \(V(e)\) given by (\ref{eq2.9}) is decreasing in \(\kappa\).

	For 2, let us first consider the case of \(M(e)=M>0\). Notice that \(\sqrt{\Delta}\) is bounded from above by \(\kappa\) and from below by \(\frac{\gamma-1}{\gamma+1}\kappa\). Applying the lower bound,
	\begin{equation*}
		\frac{p-\alpha}{q-\alpha} > \frac{(\gamma-1)^2}{\gamma+1}\kappa+\frac{2\gamma}{\gamma+1}.
	\end{equation*}
	Applying both bounds and omitting the negative term containing \(\frac{1}{\kappa}\),
	\begin{equation} \label{eq.c6}
		\begin{aligned}
			\frac{\partial}{\partial\gamma}\left(\frac{p-\alpha}{q-\alpha}\right)
			&= \frac{1}{2}\kappa-\frac{2}{(\gamma+1)^2}(\kappa-1)+\frac{1}{2}\sqrt{\Delta}+\frac{1}{\sqrt{\Delta}}\frac{\gamma-1}{(\gamma+1)^3}(\kappa-1) \\
			&< \frac{\gamma^2+2\gamma-1}{(\gamma+1)^2}\kappa + \frac{3}{(\gamma+1)^2}.
		\end{aligned}
	\end{equation}
	Hence,
	\begin{equation} \label{eq.c7}
		\begin{aligned}
			&\frac{\partial}{\partial\gamma}\left[\frac{p-\alpha}{q-\alpha}\left(\frac{q}{p}\right)^{M+1}\right] \\
			=& \frac{1}{\gamma^{M+1}}\left[\frac{\partial}{\partial\gamma}\left(\frac{p-\alpha}{q-\alpha}\right)-(M+1)\frac{p-\alpha}{q-\alpha}\frac{1}{\gamma}\right] \\
			<& \frac{1}{\gamma^{M+1}}\left[\frac{\gamma^2+2\gamma-1}{(\gamma+1)^2}\kappa + \frac{3}{(\gamma+1)^2} - (M+1)\left(\frac{(\gamma-1)^2}{\gamma(\gamma+1)}\kappa+\frac{2}{\gamma+1}\right)\right] \\
			=& \frac{[-M\gamma^3+(M+3)\gamma^2+M\gamma-(M+1)]\kappa-\gamma[2(M+1)(\gamma+1)-3]}{\gamma^{M+2}(\gamma+1)^2}.
		\end{aligned}
	\end{equation}
	Notice that \(-M\gamma^3+(M+3)\gamma^2+M\gamma-(M+1)\) has three zeros: \(\gamma_1<0<\gamma_2<1<\gamma_3\). When \(\gamma\geq\gamma_3\), the last line of (\ref{eq.c7}) is negative for all \(\kappa>1\). By continuity, there exists \(\hat{\gamma}\in[1,\gamma_3]\) such that when \(\gamma\geq\hat{\gamma}\), the first line of (\ref{eq.c7}) is negative, so \(V(e)\) given by (\ref{eq2.9}) is increasing in \(\gamma\). Notice that \(\gamma_3\) is strictly decreasing in \(N\) and converges to 1 as \(M\to\infty\). Therefore, \(\hat{\gamma}\to1\) as \(M\to\infty\). Lastly, we observe that \(\hat{\gamma}>1\). By (\ref{eq.c4}), \(\frac{p-\alpha}{q-\alpha}\to1\) as \(\gamma\downarrow 1\); by the first line of (\ref{eq.c6}), \(\frac{\partial}{\partial\gamma}\left(\frac{p-\alpha}{q-\alpha}\right)\to\frac{1}{2}+\frac{1}{2}\sqrt{2\kappa-1}\) as \(\gamma\downarrow 1\); therefore, by the first equality in (\ref{eq.c7}),
	\begin{equation} \label{eq.c8}
		\frac{\partial}{\partial\gamma}\left[\frac{p-\alpha}{q-\alpha}\left(\frac{q}{p}\right)^{M+1}\right] \to \frac{1}{2}\sqrt{2\kappa-1}-\frac{1}{2} - M
	\end{equation}
	as \(\gamma\downarrow 1\). For \(\kappa>2M^2+2M+1\), the right hand side of (\ref{eq.c8}) is positive. By continuity, there is a small neighborhood to the right of 1 where \(V(e)\) is not increasing in \(\gamma\) for all \(\kappa\), so \(\hat{\gamma}>1\).

	For the case of \(M(e)=0\), we compute the following:
	\begin{align}
		\frac{\partial}{\partial p}\left[\frac{p-\alpha}{q-\alpha}\left(\frac{q}{p}\right)^{M+1}\right] &= \left(\frac{q}{p}\right)^{M+1}\frac{1}{q-\alpha}\left[1+\frac{p-q}{q-\alpha}\frac{q}{1-2\alpha}-(M+1)\frac{p-\alpha}{p}\right], \label{eq.c9}\\
		\frac{\partial}{\partial q}\left[\frac{p-\alpha}{q-\alpha}\left(\frac{q}{p}\right)^{M+1}\right] &= \left(\frac{q}{p}\right)^{M+1}\frac{1}{q-\alpha}\left[(M+1)\frac{p-\alpha}{q} -\frac{p-\alpha}{q-\alpha} - \frac{p-q}{q-\alpha}\frac{p}{1-2\alpha}\right]. \label{eq.c10}
	\end{align}
	When \(M=0\), the bracket on the right hand side of (\ref{eq.c9}) is positive, and the bracket on the right hand side of (\ref{eq.c10}) is negative. Hence, if \(M(e)=0\), \(V(e)\) is decreasing in \(p\) and increasing in \(q\), so it is decreasing in \(\gamma\).

	We continue to show 3. Notice that the bracket on the right hand side of (\ref{eq.c9}) is increasing in \(p\) and decreasing in \(M\). Hence, there exists \(\hat{p}\in[q,1-q]\) such that (\ref{eq.c9}) is negative when \(p<\hat{p}\) and positive when \(p>\hat{p}\). Moreover, \(\hat{p}\) is non-decreasing in \(M\). That is, \(V(e)\) is increasing in \(p\) when \(p<\hat{p}\) and decreasing in \(p\) when \(p>\hat{p}\). Finally, notice that as \(p\uparrow 1-q\), \(\kappa\to\infty\). Hence, by (\ref{eq.c4}), \(\frac{p-\alpha}{q-\alpha}\to\infty\), and \(V(e)\to 0\). Therefore, \(V(e)\) is decreasing in \(p\) on a small neighborhood to the left of \(1-q\), so \(\hat{p}<1-q\).
\end{proof}

\bibliography{Disclosure_Games}

\end{document}